\documentclass{article}
\usepackage{amsmath,amssymb}
\usepackage{amsthm} 
\usepackage{xspace}
\usepackage{xcolor} 
\usepackage{xstring} 
\usepackage{graphicx} 
\usepackage{multirow} 
\usepackage{array} 
\usepackage{fullpage}
\usepackage{subcaption} 
\usepackage{authblk}
\usepackage{hyperref}
\usepackage{cleveref}
\usepackage{listings}
\definecolor{codegray}{gray}{0.95}
\providecommand{\ifstrempty}[3]{\begingroup\edef\temp{#1}\ifx\temp\empty #2\else #3\fi\endgroup}
\providecommand{\ifstrequal}[4]{\def\tempa{#1}\def\tempb{#2}\ifx\tempa\tempb #3\else #4\fi}

\theoremstyle{definition}
\newtheorem{definition}{Definition}
\theoremstyle{plain}
\newtheorem{lemma}{Lemma}
\newtheorem{theorem}{Theorem}

\newcommand{\stgraph}{$s$-$t$ graph\xspace}
\newcommand{\stgraphs}{$s$-$t$ graphs\xspace}
\newcommand{\stwalk}{$s$-$t$ walk\xspace}
\newcommand{\stwalks}{$s$-$t$ walks\xspace}

\newcommand{\stpaths}{$s$-$t$ paths\xspace}

\newcommand{\flowpaths}{\texttt{flowpaths}\xspace}
\newcommand{\highs}{\texttt{HiGHS}\xspace}
\newcommand{\ecoli}{\textbf{ecoli}\xspace}
\newcommand{\complex}{\textbf{complex32}\xspace}
\newcommand{\medium}{\textbf{medium20}\xspace}
\newcommand{\jgi}{\textbf{JGI}\xspace}

\newcommand{\ILPRobust}{MinPathError\xspace}
\newcommand{\ILPLAE}{LeastAbsErrors\xspace}
\newcommand{\ILPMFD}{MinFlowDecomp\xspace}

\newcommand{\ext}[1]{\mathsf{extension}(#1)}

\makeatletter
\newcommand{\unmarkedfootnote}[1]{%
  \begingroup
    \let\@makefnmark\relax
    \long\def\@makefntext##1{\parindent 1em\noindent ##1}%
    \footnotetext{#1}%
    \addtocounter{footnote}{-1}
  \endgroup
}
\makeatother

\newcommand{\rev}[1]{\textcolor{black}{#1}}

\title{Fast and Flexible Flow Decompositions in General Graphs\\ via Dominators} 

\author[1]{Francisco Sena}
\author[1]{Alexandru I.~Tomescu}
\affil[1]{Department of Computer Science, University of Helsinki, Finland\\
\texttt{\{alexandru.tomescu,francisco.sena\}@helsinki.fi}} 

\date{} 

\begin{document}

\maketitle
\unmarkedfootnote{Co-funded by the European Union (ERC, SCALEBIO, 101169716). Views and opinions expressed are however those of the author(s) only and do not necessarily reflect those of the European Union or the European Research Council. Neither the European Union nor the granting authority can be held responsible for them. This work has also received funding from the Research Council of Finland grants No.~346968, 358744.\\[0.2cm]\includegraphics[width=4cm]{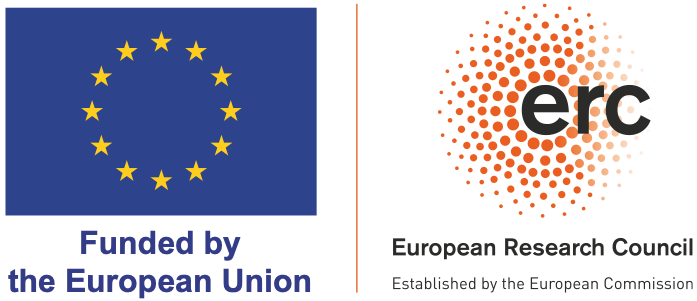}}

\begin{abstract}
Multi-assembly methods rely at their core on a flow decomposition problem, namely, decomposing a weighted graph into weighted paths or walks. However, most results over the past decade have focused on decompositions over directed acyclic graphs (DAGs). This limitation has lead to either purely heuristic methods, or in applications transforming a graph with cycles into a DAG via preprocessing heuristics. In this paper we show that flow decomposition problems can be solved in practice also on general graphs with cycles, via a framework that yields fast and flexible Mixed Integer Linear Programming (MILP) formulations.

Our key technique relies on the graph-theoretic notion of \emph{dominator tree}, which we use to find all \emph{safe sequences of edges}, that are guaranteed to appear in some walk of any flow decomposition. We generalize previous results from DAGs to cyclic graphs, by showing that maximal safe sequences correspond to extensions of common leaves of two dominator trees, and that we can find all of them in time linear in their size.

Using these, we can accelerate MILPs for \emph{any} flow decomposition into walks in general graphs, by setting to (at least) 1 suitable variables encoding solution walks, and by setting to 0 other walks variables non-reachable to and from safe sequences. This reduces model size and eliminates costly linearizations of MILP variable products. 

We experiment with three decomposition models (Minimum Flow Decomposition, Least Absolute Errors and Minimum Path Error), on four bacterial datasets. Our pre-processing enables up to thousand-fold speedups and solves even under 30 seconds many instances otherwise timing out. We thus hope that our dominator-based MILP simplification framework, and the accompanying software library can become building blocks in multi-assembly applications.
\end{abstract}

\medskip
\noindent\textbf{Keywords:} directed graph, walk cover, dominator tree, graph algorithm, graph reachability, integer linear programming, flow decomposition, optimization problem, multi-assembly
\thispagestyle{empty}
\setcounter{page}{0}
\newpage

\section{Introduction}

\paragraph{Background.} Decomposing a weighted graph into weighted paths is a basic computational problem at the core of several types of \emph{multi-assembly problems}, for example RNA transcript assembly~\cite{tomescu2015explaining,translig,cliiq,multitrans,ssp,nsmap,jumper,class2,isoinfer,dias2022fast,dias2023safety,scallop2,WilliamsTM23,isoquant}, metagenomic assembly~\cite{namiki2011metavelvet,peng2011meta,li2015megahit,nurk2017metaspades,kolmogorov2020metaflye,heber2002splicing,vicedomini2021strainberry,sena2024flowtigs} or the assembly of viral strains~\cite{topfer2014viral,baaijens2020strain,baaijens2019full,BeerenwinkelBBD15,luo2022strainline,tian2025accurate,freire2022viquf,luo2023vstrains}. In these kind of problems, multiple (similar) genomic sequences, in different abundances in a sample, need to be recovered from mixed reads sequenced from all of them.

In popular models of this problem, the reads give rise to a weighted graph, whose vertices encode partial genomic sequences, and whose edges encode observed connections between them (e.g.~reads spanning two such genomic sequences). In addition, vertices or edges (or both) have weights corresponding to their read coverage. The genomic sequences to be recovered correspond to \emph{weighted paths} in this graph, the weight being their abundance in the sample. Therefore, to solve the initial multi-assembly problem, one needs to ``optimally'' decompose the weighted graph into weighted paths.

For directed acyclic graphs (DAGs), there is a wide body of work on this problem. For example, various formulations of optimality for a set of weighted paths exist~\cite{tomescu2015explaining,translig,cliiq,multitrans,ssp,nsmap,jumper,class2,isoinfer,dias2022fast,dias2023safety,dias2024robust,stinzendorfer2024robust}. Among these, the most basic is the Minimum Flow Decomposition (MFD) problem for error-free data~\cite{pop2009genome,shao2017theory,kloster2018practical}, where the input is a DAG with a weight on every edge (its \emph{flow value}), satisfying \emph{flow conservation} at every vertex except sources or sinks (i.e.~the sum of flow values on incoming edges equals that on outgoing edges). The MFD problem asks for a minimum number of weighted paths such that the flow value of every edge equals the sum of weights of paths passing through it~\cite{vatinlen2008simple}. MFD is NP-hard~\cite{vatinlen2008simple}, even on simple DAG classes~\cite{hartman2012split,grigorjew2024parameterisedapproximation}, and hence most multi-assembly formulations for real, imperfect data, are NP-hard~\cite{tomescu2015explaining}.

\paragraph{Motivation.} The acyclicity of the graph is a powerful structural property that has been exploited in numerous practical solution for MFD, despite its hardness. For example, on DAGs, MFD admits a fixed-parameter tractable algorithm~\cite{kloster2018practical}, tailored heuristics~\cite{shao2017theory,mumey2015parity,vatinlen2008simple}, approximation algorithms on some classes of DAGs~\cite{caceres2024width}, or integer linear programs (ILPs)~\cite{dias2022fast}, which further admit optimizations to significantly reduce the runtime of ILP solvers~\cite{acceleratingILP}.

The situation is radically different for general graphs that may contain cycles, where there are far fewer methods to approach MFD: two heuristics appear in~\cite{vatinlen2008simple}, an approximation algorithm working on a certain class of graphs was proposed by~\cite{caceres2024width}, and an exact ILP-based method was proposed by~\cite{dias2022minimum}. Moreover, all these work only for perfect data, and thus cannot be applied to real data. In fact, some multi-assembly methods employ various strategies or construction parameters to make a graph built from real data acyclic~\cite{petri2023isonform,baaijens2019full}.

Moreover, we have shown that on DAGs, \emph{any} flow decomposition problem for erroneous data admits fast ILP-based solvers~\cite{sena2025safe}, not only MFD. 
This methodology works by (1) abstracting the decomposition problem as one whose solution consists of a set of source-to-sink paths that cover (a given subset of) the vertices/edges (they form a \emph{path cover} of the DAG); (2) characterizing the sequences that must appear in a path of \emph{every} such path cover (\emph{safe sequences}); (3) developing a linear-time algorithm finding all maximal safe sequences in a DAG, based on the well-known notion of \emph{dominator} (see e.g.~\cite{parotsidis2013dominators}); and (4) reducing the search space of the ILP by fixing to 1 suitable binary variables corresponding to incompatible safe sequences.

\paragraph{Contributions.} In this paper we show that this methodology is not limited to DAGs, but can be adapted also to general graphs that may contain cycles. For (1), we consider decomposition problems whose solution consists of a set of source-to-sink \emph{walks} that cover (a given subset of) the vertices/edges (a \emph{walk cover}--walks may repeat vertices and edges). For (2) and (3) we show that the techniques developed for the DAG case can be extended to the general case of graphs with cycles; we also obtain a linear-time algorithm outputting all maximal safe sequences of walk covers, also based on dominators. For (4) we use the condensation DAG of strongly connected components to find maximal safe sequences that can be guaranteed to appear in different solution walks. Based on them, we use again the condensation DAG to generalize the fixing of walk variables of the ILP. As a novelty, in this paper we also fix to 0 some integer variables corresponding to edges that cannot be traversed by a solution path. Additionally, by this fixing to 0 or 1, we can also avoid the expensive linearization of products between two integer variables, leading to smaller models that are faster to solve.

Furthermore, for (2) and (3), we introduce the notion of \emph{collapsing} paths on the dominator trees; the dominance relation of the graph is invariant under this operation. Importantly, this concept gives a unified view of walk covers where every vertex/edge or only a subset thereof has to be covered. This approach greatly simplifies the techniques developed in~\cite{sena2025safe} (e.g., ``compressing'' the input DAG), and, in fact, those techniques can be seen as a special application of the techniques developed in this paper.

As a base ILP formulation of source-to-sink walks in general graphs, we use the one from~\cite{dias2022minimum}, which was proposed for the MFD problem. Here we slightly simplify it, and we extend it also to decomposition problems applicable to imperfect data. As two concrete ILP models for such data, in this paper we consider the \ILPLAE model~\cite{baaijens2020strain,traph,bernard2014efficient} (where we need to minimize the total sum of the absolute differences between the weight of every edge and the sum of the weights of the walks using it), and a more complex and more accurate ``robust'' model, \ILPRobust~\cite{dias2024robust} (where paths also get a slack variable, whose total sum needs to be minimized, and this absolute difference for every edge must be below the sum of the slacks of the walks passing through it). We define these models formally in~\Cref{sec:notation}. 

Moreover, in applications of flow decomposition problems, more information is often available than just the graph structure and the edge weights. For example, long reads overlapping more than one edge indicate that these edges must co-appear in some solution path or walk~\cite{BeerenwinkelBBD15,scallop2}. For acyclic graphs, these have been modeled previously as \emph{subpath constraints}~\cite{WilliamsTM23,RizziTM14,acceleratingILP,dias2022fast}, namely given paths that must appear as a subpath of at least one solution path of the decomposition problem. For general graphs with cycles, we introduce here a weaker form of this constraint (but which can be directly added to the MILP models above), in the form of a \emph{subset constraint}, namely a set $S$ of edges that must appear in at least one walk of any solution.

We implemented all our ILP models and their optimizations inside the \flowpaths Python package for flow decomposition problems (\url{https://algbio.github.io/flowpaths}). This package uses by default the popular open source MILP solver \highs, see \url{www.highs.dev}, which is installable without a license, thus maximizing the potential for applications in bioinformatics software. The architecture of the \flowpaths package is designed so that the formulation of paths is inside an abstract class, which is then instantiated by any flow decomposition model adding its own problem-specific constraints. We followed the same approach for walks: the optimizations proposed here and the subset constraints are all implemented in an abstract class, and thus applicable to any other decomposition model into walks that might be developed in the future based on \flowpaths. Given the generality of this software design, we focused our experiments on measuring the speed-ups that our optimizations can give on input graphs created from different bacterial datasets. We observe speed-ups that grow with the model complexity (from \ILPMFD, to \ILPLAE, to \ILPRobust), and likewise grow with the number of bacterial genomes that are in the graph. On the hardest datasets that we tested we obtain speed-ups of 393$\times$ for \ILPMFD, 560$\times$ for \ILPLAE, and 1465$\times$ for \ILPRobust, and solving many instances in under 30 seconds on average that would otherwise time out.

\section{Notation and preliminaries}
\label{sec:notation}

\paragraph{Graphs.} An \emph{\stgraph} is a directed graph without parallel edges (but possibly with self-loops), with a unique source $s$ (a vertex with no in-coming edges) and a unique sink $t$ (a vertex with no out-going edges). Let $G = (V,E)$ be an \stgraph. A \emph{walk} in $G$ is a sequence of vertices $v_1v_2\ldots v_t$ such that $v_iv_{i+1} \in E$ for all $1 \leq i < t$; note that a walk can revisit vertices and edges. For designated vertices $u, v \in V$, we say that a walk $W$ in $G$ is an \emph{$u$-$v$ walk} if the first vertex is $u$ and the last vertex is $v$. If the walk does not repeat vertices (and thus also edges), we say that it is a \emph{path}. For an edge $uv \in E$, we denote by $W(uv)$ the number of times the edge $uv$ is traversed by the walk $W$; if $W$ does not traverse $uv$ then we set $W(uv) = 0$. The \emph{distance} between two vertices is the number of edges in a shortest path between them.

\paragraph{Sequences and safety.} Let $C$ be a subset of vertices/edges of $G$. A \emph{$C$-walk cover} of $G$ is a set $P$ of $s$-$t$ walks such that for every vertex/edge $u \in C$ there is a walk in $P$ containing $u$. We denote a \emph{sequence} $X$ through vertices/edges $u_1,\dots,u_\ell$ as $X=u_1,\dots,u_\ell$ or $X=(u_1,\dots,u_\ell)$ if there is a $u_i$-$u_{i+1}$ path for all $1 \le i \le \ell-1$. Any path is a sequence. We say that $X$ \emph{contains} a vertex $v$ if $v=u_i$ for some $i \in \{1, \dots, \ell\}$. The internal vertices of a path are the vertices contained in the path except its first and last vertices. 
Let $X'=v_1,v_2,\dots,v_{\ell'}$ be a sequence. If $X$ can be obtained from $X'$ by deleting any number of vertices/edges then $X$ is a \emph{subsequence} of $X'$ (and $X'$ is a \emph{supersequence} of $X$).
Suppose that there is a path from the last vertex $u_\ell$ of $X$ to the first vertex $v_1$ of $X'$. The \emph{concatenation} of $X$ and $X'$ is the sequence $XX'$ obtained as $X$ followed by $X'$; if $u_\ell=v_1$, then the concatenation is $XX' := u_1,\dots,u_\ell,v_2,\dots,v_{\ell'}$, i.e. the repeated occurrence of $u_\ell=v_1$ is removed. The definitions above can be adapted in the obvious way for sequences of edges.
Let $C$ be a subset of $V$ (resp.~$E$) and $X$ be a sequence of vertices (resp.~edges) of $G$. We say that $X$ is a \emph{$C$-safe sequence} if for any $C$-walk cover $P$ of $G$, there is a walk $W$ in $P$ such that $X$ is a subsequence of $W$ (when clear from the context, we omit the prefix $C$-). See \Cref{fig:safe-sequences-nodes} for an illustration of four safe sequences when $C = V$.

\paragraph{Dominators.} We say that vertex $u$ \emph{$s$-dominates} vertex $v$ if every $s$-$v$ path contains $u$. Every vertex $s$-dominates itself.
Vertex $u$ \emph{strictly $s$-dominates} $v$ if $u$ $s$-dominates $v$ and $u \neq v$. For $v \neq s$, the \emph{immediate $s$-dominator} of $v$ is the vertex $u \neq v$ that $s$-dominates $v$ and is $s$-dominated by all the vertices that strictly $s$-dominate $v$; we write $idom_s(v)=u$.
The domination relation on $G$ with respect to $s$ can be represented as a tree rooted at $s$ with the same vertex set as $G$, called the \emph{$s$-dominator tree} of $G$. In this tree, every vertex has as its parent its immediate $s$-dominator, and every vertex is $s$-dominated by all its ancestors and $s$-dominates all its descendants. A vertex is said to be an \emph{$s$-dominator} if it strictly $s$-dominates some vertex.
For the domination relation with respect to $t$, it is enough to say that a vertex $u$ \emph{$t$-dominates} vertex $v$ if every $v$-$t$ path contains $u$; the remaining definitions are completely analogous. See \Cref{fig:dominator-example} for an illustration and~\cite{parotsidis2013dominators} for a more in-depth presentation of dominators.

It is convenient to consider the dominator trees as directed trees, so that the edges of the $s$-dominator tree point away from the root and the edges of the $t$-dominator tree point towards the root. In this way a path in the $s$-dominator tree goes from root-to-vertex and a path in the $t$-dominator tree from vertex-to-root.
Additionally, as to generalize the immediate-dominance relation, we define a function $dom: V \times \mathbb{N}^+ \to V$ on dominator trees as follows. Let $dom_s(v,k) := idom_s(v)$ when $k=1$ and $dom_s(v,k) := dom_s(dom_s(v,1),k-1)$ when $k>1$. If $k$ is larger than the number of \emph{strict} dominators of a given vertex, then it defaults to $s$ (resp.~$t$) in the $s$-dominator tree (resp. $t$-dominator tree). Essentially, $dom$ gives the $k$-th ancestor of a vertex in a rooted tree, if it is well defined, otherwise it returns the root of the tree.
By $\ext{v}$ we mean the sequence of vertices obtained by concatenating the path from $s$ to $v$ in the $s$-dominator tree with the path from $v$ to $t$ in the $t$-dominator tree. The \emph{depth} of a vertex $v$ in a rooted tree is the distance between $v$ and the root of the tree.


\paragraph{Flow decomposition models.} Next, we formally define the three flow decomposition problem variants that we address in the paper.

\begin{definition}[Flow decomposition variants] 
    \label{def:problems}
    Let $G = (V,E)$ be an \stgraph, let $E' \subseteq E$ be a subset of edges, let $f : E \rightarrow \mathbb{Z}_{\geq0}$ be a function assigning a non-negative integer weight to every edge of $G$, and let $k \geq 1$ be an integer. We define below three problems that require finding $k$ 
    \stwalks $W_1,\dots,W_k$ in $G$, with associated weights $w_1,\dots,w_k \in \mathbb{Z}^{+}$, respectively, with the following objectives:
    \begin{description}
        \item[$k$-Flow Decomposition ($k$-FD) problem] For every edge $uv \in E'$, we have: $\sum_{i=1}^k W_i(uv)w_i = f(uv)$.
        \item[$k$-Least Absolute Errors ($k$-LAE) problem] The walks and their associated weights minimize: \\$\sum_{uv \in E'} \left|f(uv) - \sum_{i=1}^k W_i(uv)w_i\right|$.
        \item[$k$-Min Path Error ($k$-MPE) problem] Find also an associated slack $\rho_i \in \mathbb{Z}_{\geq0}$ for each walk $W_i$, such that: $\displaystyle \left|f(uv) - \sum_{i=1}^k W_i(uv)w_i\right| \leq \sum_{i=1}^{k} W_i(uv)\rho_i$ holds for all $uv \in E'$, and minimizing $\sum_{i=1}^{k} \rho_i$.
    \end{description}
\end{definition}

Note that in the above we assumed that all walks start in the same vertex $s$ and end in the same vertex $t$. However, in practice we might have a set $S$ of vertices where the walks are allowed to start, especially if the graph has more sources, and likewise for a set $T$ where the walks are allowed to end. However, allowing for such $S$-$T$ walks is easy, because we considered the generalization in which the objectives are applied only to the subset of edges $E'$. We can then simply add a new vertex $s$ to such an input, together with edges from $s$ to every vertex in $S$, and a new vertex $t$ together with edges from every vertex in $T$ to $t$. These new edges incident to $s$ or $t$ are not added to $E'$, hence they can get weight 0 and they do not interfere with the problem objectives (which are applied to $E'$).

Due to lack of space, we describe the MILP formulations for the problems from \Cref{def:problems} in \Cref{sec:milp-formulations}. For the rest of the paper, it suffices to know that every walk $W_i$, with $i \in \{1,\dots,k\}$ is modeled by variables $x_{uv,i} \in \mathbb{Z}_+$ for every edge $uv \in E$, such that $x_{uv,i} = W_i(uv)$.

\section{Safe sequences and dominators}

Notice that the flow decomposition models from~\Cref{def:problems} are defined so that their solution walks cover the edges of the graph, and so we are interested in safe sequences for $C \subseteq E$. However, dominator trees are usually defined in terms of vertices. Thus, we state the results in terms of $C \subseteq V$ as they may also be of independent interest outside of bioinformatics applications.
We first solve the case $C=V$ and then show how to generalize for $C \subseteq V$ (the former results serve as preliminary steps toward understanding the more general theorems). Further, we do not give proofs of the statements $C=V$ as they follow as a particular case of the results for $C \subseteq V$.
The case $C \subseteq E$ can be handled in two manners: either (1) introduce a new vertex $v_e$ in the middle of every edge $e \in C$, add all such $v_e$ to $C'$, and compute all $C'$-safe sequences (notice that the new graph has only $|E|$ more vertices and $|E|$ more edges than the original graph); or (2) obtain direct analogues of the vertex-characterizations in terms of edges via the line graph of $G$ and then use dominator trees with respect to the edge-dominance relation for the enumeration algorithm (for a running example of this approach, observe that \Cref{fig:safe-sequences-nodes} is the line graph of \Cref{fig:antichain-2} and compare the safe sequences shown with matching colors, the former figure w.r.t. vertices and the latter w.r.t. edges). The relevant missing proofs from this section can be found in~\Cref{apx:missing-proofs}.

\begin{theorem}\label{thm:safe-sequences-characterization}
    Let $G=(V,E)$ be an $s$-$t$ graph. A sequence $X$ of vertices is safe for walk covers if and only if there exists a vertex $v \in V$ such that $X$ is a subsequence of $\ext{v}$.
\end{theorem}

By~\Cref{thm:safe-sequences-characterization} it follows that every maximal safe sequence is the extension of some vertex. Thus a graph with $n$ vertices has at most $n$ maximal safe sequences, none of which skips vertices on the corresponding paths in the dominator trees. Notice also that any vertex appears at most twice in any safe sequence $X$, since a vertex appears at most once in the sequence of the $s$-dominators and $t$-dominators of some vertex.

The next lemma was first shown in~\cite{sena2025safe}, and essentially relates the $s$- and $t$-dominance relations. The result is merely technical and is required for the characterization of maximal safe sequences.

\begin{lemma}
\label{lem:st-dom-immediate-relation}
    Let $G=(V,E)$ be an $s$-$t$ graph, let $u,v \in V$ be vertices, and let $k \in \mathbb{N}^+$. If $u$ is the $k$-th ancestor of $v$ in the $s$-dominator tree, then $dom_t(u,k)$ $t$-dominates $v$, and $v$ is not a $t$-dominator of $u$ unless $v=dom_t(u,k)$.
\end{lemma}

To obtain a characterization of maximal safe sequences on DAGs,~\cite{sena2025safe} assumes the graphs to be ``compressed'', i.e., that every path in the original graph where every vertex but the first has in-degree one and every vertex but the last one has out-degree one is represented as a single vertex. (In bioinformatics these paths are known as ``unitigs'', and \cite{sena2025safe} refers to them as ``unitary paths''.) This causes the characterization of maximal safe sequences with respect to $C$-walk covers to be more involved than necessary.
Instead, we propose to collapse particular paths on the dominator trees, and, as we will show, this operation preserves the set of maximal safe sequences of the graph. Moreover, the paths compressed by~\cite{sena2025safe} always correspond to these particular paths in the dominator trees. The converse, however, does not hold due to the presence of cycles in our graph,\footnote{Notice the path $uv$ in~\Cref{fig:subset-cover}. Another example is the univocal path $abt$ in the graph with edge list $sa,sx,ab,bx,xa,bt$.} and hence our approach indeed generalizes that of~\cite{sena2025safe}.

\begin{definition}[Univocal path\footnote{It is not hard to see that this kind of sequence of dominators indeed form a path of $G$.}]
   Let $G$ be an $s$-$t$ graph. Let $p=v_1 \dots v_k$, $k\geq 1$, be a path in the $t$-dominator tree such that $v_{i+1}$ has a unique child $v_i$ for $i=1,\dots,k-1$. If $v_1 \dots v_k$ is a path in the $s$-dominator tree such that $v_{i}$ has a unique child $v_{i+1}$ for $i=1,\dots,k-1$, then $p$ is a \emph{univocal path} of $G$.
\end{definition}

\begin{figure}[t]
    \begin{subfigure}[m]{0.32\textwidth}
        \centering
        \includegraphics[width=\textwidth]{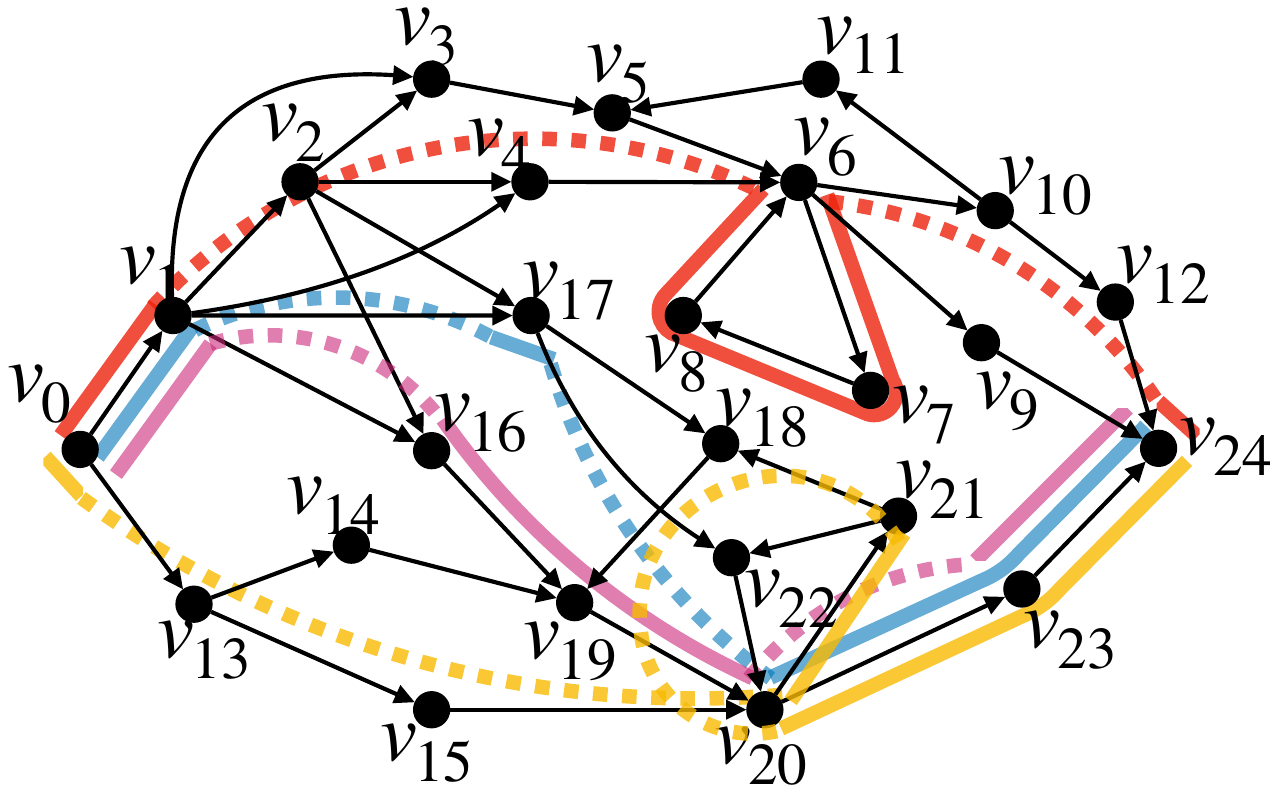}    
        \caption{An \stgraph $G$, $s=v_0$, $t=v_{24}$\label{fig:safe-sequences-nodes}}
    \end{subfigure}
    \hfill
    \begin{subfigure}[m]{0.32\textwidth}
        \centering
        \includegraphics[width=\textwidth]{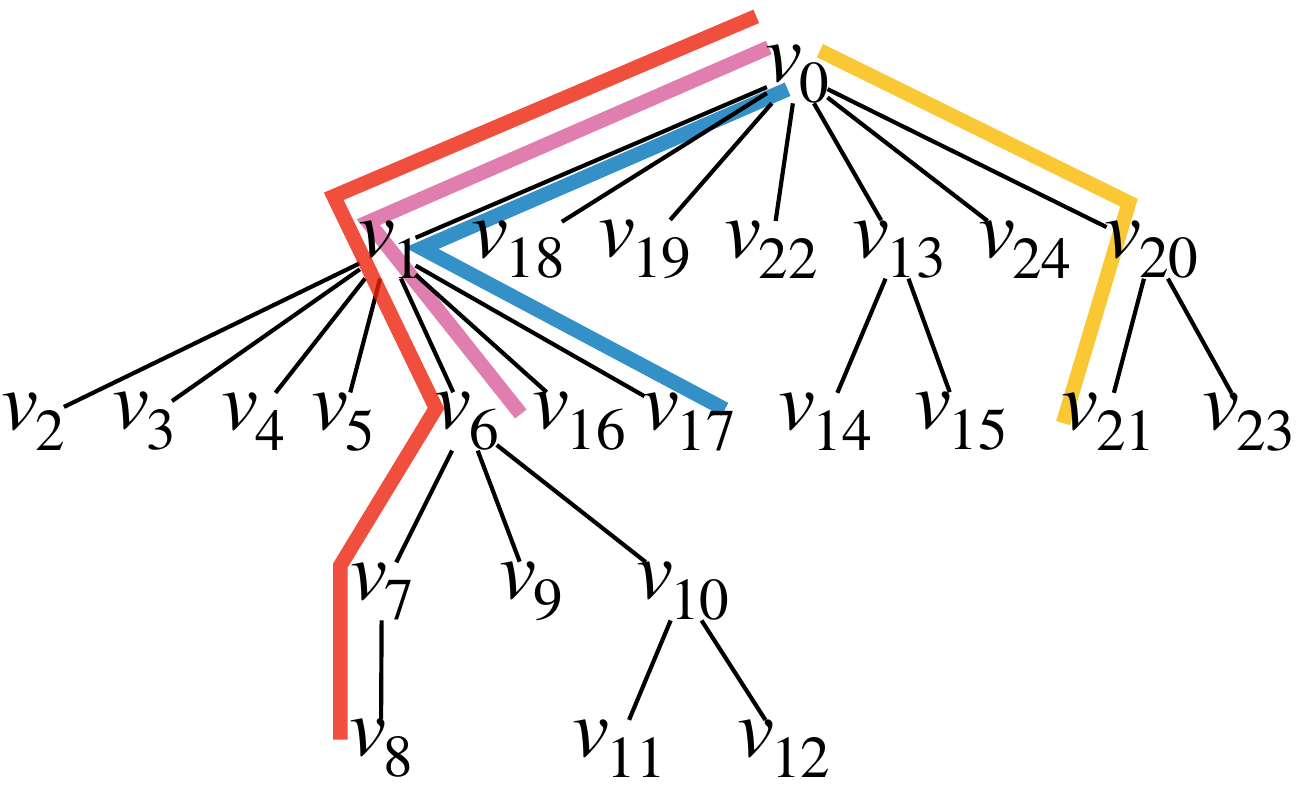}    
        \caption{The $s$-dominator tree of $G$\label{fig:s-dominator-tree}}
    \end{subfigure}
    \hfill
    \begin{subfigure}[m]{0.32\textwidth}
        \centering
        \includegraphics[width=\textwidth]{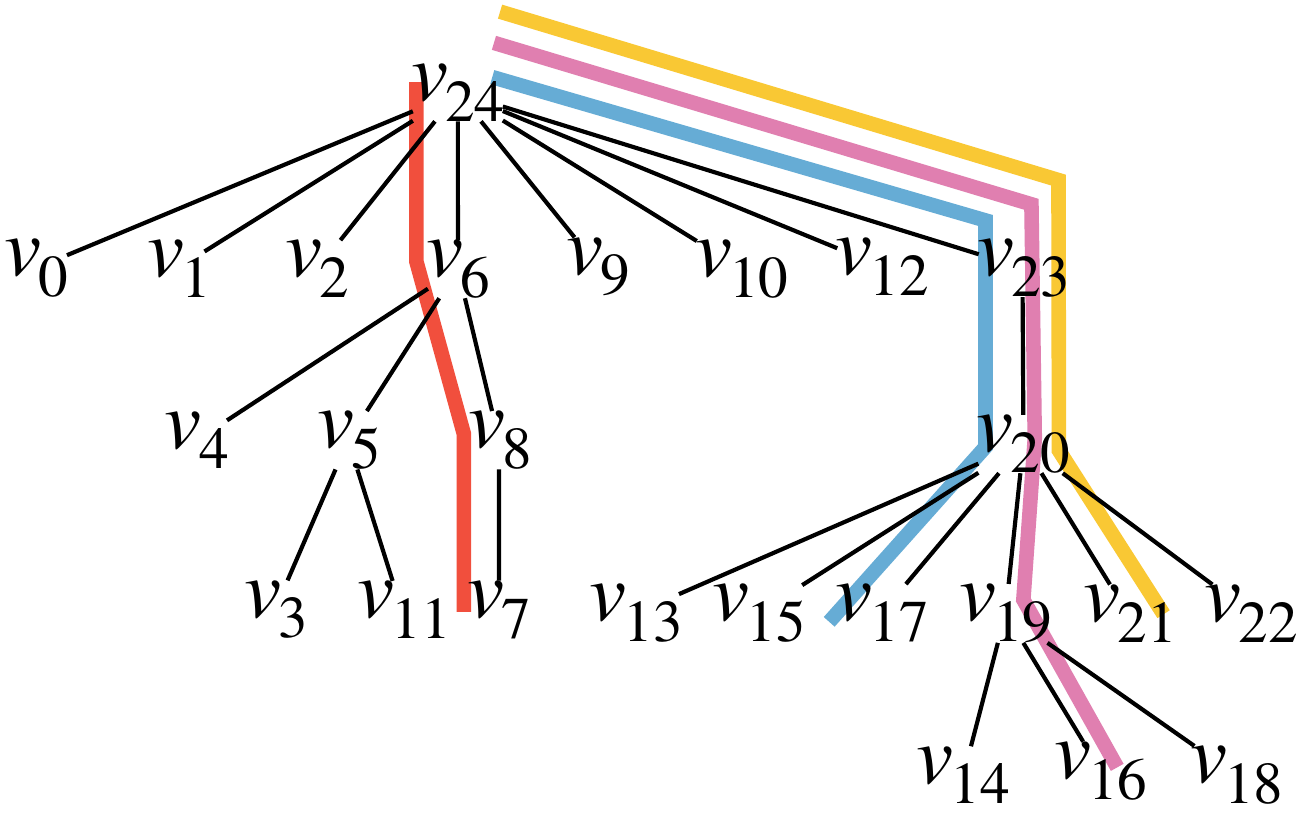}    
        \caption{The $t$-dominator tree of $G$}
    \end{subfigure}
    \caption{Example of a graph $G$, its $s$- and $t$-dominator trees, four $C$-safe sequences in $G$ with $C = V$. For illustration purposes, in a sequence we draw univocal edges between nodes as solid lines, and other connections as dashed: for example, the yellow safe sequence is $(v_0,v_{20},v_{21},v_{20},v_{23},v_{24})$. Because of \Cref{thm:cores-vertices}, maximal safe sequences are obtained by concatenating the paths in the $s$- and $t$-dominator trees: for example, the yellow sequence is obtained by concatenating the path $v_0v_{20}v_{21}$ in the $s$-dominator tree with the path $v_{21}v_{20}v_{23}v_{24}$ in the $t$-dominator tree.
    \label{fig:dominator-example}}
\end{figure}

Let $p=v_1\dots v_k$ be a maximal univocal path of an \stgraph $G$ and let $T_s$ and $T_t$ denote the $s$-dominator and $t$-dominator trees of $G$, respectively $(k\geq 1)$.
By definition of univocal path the extensions of the vertices of $p$ all yield the same sequence.
Importantly, the extension of any vertex $v$ not in $p$ either contains $p$ or does not contain any vertex of $p$: if $v$ is a descendant of $v_1$ in $T_t$ or of $v_k$ in $T_s$ then $\ext{v}$ clearly contains $p$, and if $v$ is not a descendant of $v_1$ in $T_t$ and is not a descendant of $v_k$ in $T_s$, then the lowest common ancestor of $v$ and $v_1$ (resp. $v_k$) in $T_t$ (resp. $T_s$) is a strict ancestor of $v_k$ (resp. $v_1$) in $T_t$ (resp. $T_s$), and thus $\ext{v}$ does not contain any vertex of $p$. Consequently, two vertices belonging to two distinct maximal univocal paths have different extensions.
We can also show that maximal univocal paths are pairwise disjoint, and so the set of maximal univocal paths uniquely partition the dominator trees into vertex-disjoint paths.

\begin{lemma}\label{lem:maximal-univocal-paths-disjoint}
    Let $G$ be an $s$-$t$ graph and let $p_1,p_2$ be two distinct maximal univocal paths. Then $p_1$ and $p_2$ are vertex disjoint.
\end{lemma}

Collapsing each maximal univocal path into a single vertex then yields two trees with essentially the same set of extensions of the trees before the collapse; we choose the deepest vertex in the $t$-dominator tree to represent the respective collapsed path. To ensure that there is no loss of information during the process, every vertex arising from a collapse operation stores the ordered sequence of vertices of the respective collapsed path (and it suffices to do this in just one of the trees). In this way, whenever an extension passes through a collapsed vertex the sequence stored at that vertex is considered for the extension, thus mimicking the extension in the original trees.
Therefore the set of maximal safe sequences before and after the collapse is identical, i.e., collapsing maximal univocal paths defines an equivalence relation with respect to extensions.

For example, in \Cref{fig:subset-cover}, the paths $uv$ and $bc$ are maximal univocal paths and should be collapsed; the path $d$ is a trivial maximal univocal path. On the other hand, the path $df$ should not be collapsed even though $\ext{d}=\ext{f}$, since otherwise this would mistakenly encode the fact that vertex $f$ $s$-dominates vertex $e$ (note, however, that this collapse preserves the $t$-dominance relation).
After collapse, the path $uv$ in the $t$-dominator tree becomes vertex $u$ with the additional information ``$v$'', and in the $s$-dominator tree the path $vu$ becomes simply vertex $u$.
We are now ready to give a characterization of maximal safe sequences.

\begin{figure}
    \centering
    \includegraphics[width=0.75\linewidth]{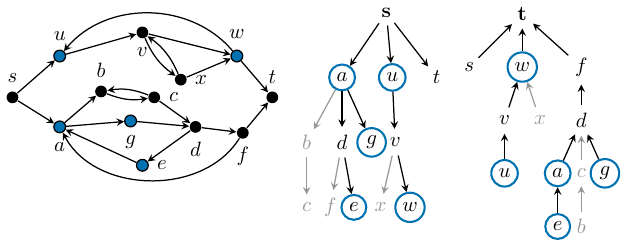}
    \caption{Example of a graph $G$, its blue-dominator trees with respect to $C = \{a,e,g,u,w\}$. Lower opacity vertices are not part of the blue-dominator trees. Underlying the blue-dominator trees are the dominator trees of $G$ for $C=V$. Vertex $e$ is a blue-child of vertex $a$ in the $s$-dominator tree (and vertex $a$ is its closest blue-ancestor). The path $uw$ is maximal $C$-univocal and its collapsed into vertex $u$, which stores the sequence $(v,w)$. The sequence $ea$ is not a $C$-univocal path because of the position of vertex $g$ relative to vertices $a$ and $e$ in the $s$-dominator tree. The maximal $C$-safe sequences are: $\ext{u}=\ext{w}=(s,u,v,w,t)$, $\ext{e}=(s,a,d,e,a,d,f,t)$, and $\ext{g}=(s,a,g,d,f,t)$.}
    \label{fig:subset-cover}
\end{figure}

\begin{theorem}[Characterization of maximal safe sequences]
\label{thm:cores-vertices}
    Let $G=(V,E)$ be an $s$-$t$ graph and suppose that the dominator trees of $G$ are collapsed. Let $u$ be a vertex in the collapsed trees. Then $\ext{u}$ is a maximal safe sequence if and only if $u$ is a leaf in both dominator trees.
\end{theorem}

Since linear-time algorithms to build dominator trees are known, see e.g.~\cite{alstrup1999dominators}, we get the next result.

\begin{theorem}[Optimal enumeration and representation of safe sequences]
\label{thm:representation-optimal-enumeration}
    Let $G$ be an $s$-$t$ graph with $m$ edges. There is an $O(m+o)$ time algorithm outputting the set of all maximal safe sequences with no duplicates, where $o$ denotes the total length of all the maximal safe sequences.
\end{theorem}

Moreover, \Cref{thm:cores-vertices} also tells us that the dominator trees encode every maximal safe sequence of $G$ via extensions. Hence, dominator trees represent all the maximal safe sequences of $G$ in $O(n)$-space.
We dedicate the rest of this section to generalizing the previous results for walk covers where only a given subset of vertices has to be covered.

\paragraph{$C$-safe sequences.}

The characterization of $C$-safe sequences is completely analogous to \Cref{thm:safe-sequences-characterization}.

\begin{theorem}[Characterization of safe sequences for $C$-walk covers]
\label{thm:safe-sequences-characterization-C}
    Let $G=(V,E)$ be an $s$-$t$ graph, $C\subseteq V$ a set of vertices. A sequence $X$ of vertices is safe for $C$-walk covers if and only if there exists a vertex $v \in C$ such that $X$ is a subsequence of $\ext{v}$.
\end{theorem}

Our goal is to establish analogues of \Cref{thm:cores-vertices} and \Cref{thm:representation-optimal-enumeration} for $C$-walk covers.
We begin by describing univocal paths with respect to $C$-walk covers.

Mark with blue the vertices of $C$ in both dominator trees of $G$. A vertex is \emph{blue} if it is marked. A \emph{blue-child} of a vertex $u$ in the $s$-dominator tree is a blue vertex $v$ that is a descendant of $u$ and no internal vertex on the $u$-$v$ path in the $s$-dominator tree is blue (and analogously for the $t$-dominator tree).
By \Cref{thm:safe-sequences-characterization-C}, vertices that do not appear in the extension of any blue vertex can be ignored/removed from the trees. This produces the \emph{blue-dominator trees}, which essentially encode the dominance relation of $G$ restricted to $C$ while maintaining relevant non-blue vertices (i.e., vertices not contained in $C$ that appear in the extension of a blue vertex, see \Cref{fig:subset-cover}).

A $C$-univocal path in the blue-dominator trees generalizes in the natural way: if $v_1\dots v_k$ is a sequence of blue vertices in the blue $t$-dominator tree such that each vertex $v_{i+1}$ has $v_i$ as its \emph{unique} blue-child for $i=1,\dots,k-1$ $(k\geq 1)$, then $v_1,\dots,v_k$ is $C$-univocal if $v_1,\dots,v_k$ is a sequence in the $s$-dominator where $v_{i}$ has $v_{i+1}$ as its \emph{unique} blue-child for $i=1,\dots,k-1$ $(k\geq 1)$.
So \Cref{lem:maximal-univocal-paths-disjoint} generalizes for $C$-univocal paths as well as the fact that collapsing maximal $C$-univocal paths preserves the extensions of the blue-dominator trees.
Further, since \Cref{lem:st-dom-immediate-relation} suitably explains how the $s$- and $t$-dominance relation interact in arbitrary ``ancestry-dominance'' relations, the characterization of maximal $C$-safe sequences is indeed a direct generalization of \Cref{thm:cores-vertices}. (We remark that in order to prove \Cref{thm:cores-vertices} directly, application of \Cref{lem:st-dom-immediate-relation} is required only for $k=1$, i.e., with the immediate-dominance relation.)

\begin{theorem}[Characterization of maximal safe sequences for $C$-walk covers]
\label{thm:cores-vertices-C}
    Let $G=(V,E)$ be an $s$-$t$ graph, $C\subseteq V$, and suppose that the dominator trees of $G$ are collapsed with respect to $C$. Let $u\in C$ be a vertex in the collapsed trees. Then $\ext{u}$ is a maximal $C$-safe sequence if and only if $u$ is a leaf in both dominator trees.
\end{theorem}

To compute maximal $C$-safe sequences we can proceed essentially as when $C=V$. The description of this procedure is given in the proof of the next theorem. Moreover, this result also implies that blue-dominator trees represent every maximal $C$-safe sequence of $G$ in $O(n)$-space.

\begin{theorem}[Maximal safe sequence enumeration for $C$-walk covers]
\label{thm:representation-optimal-enumeration-C}
    Let $G=(V,E)$ be an $s$-$t$ graph with $m$ edges and let $C \subseteq V$. There is an $O(m+o)$ time algorithm outputting the set of all maximal $C$-safe sequences for $C$-walk covers with no duplicates, where $o$ denotes the total length of all the maximal $C$-safe sequences.
\end{theorem}

\section{Simplifying flow decomposition MILP models via safe sequences}
\label{sec:simplifying-milp}

In~\cite{sena2025safe}, maximal safe sequences in acyclic \stgraphs were used to fix MILP variables encoding \stpaths in the following manner. For every edge $e$ of the graph, one stores a weight $w(e)$ computed as the length of a longest safe sequence containing $e$, and stores one such sequence as $S(e)$. Then, one computes an antichain of edges, i.e.~a set of edges $A = \{e_1,\dots,e_t\}$ that are pairwise unreachable, i.e.~there is no path in the graph from $e_i \in A$ to $e_j \in A$ for all distinct $i,j \in \{1,\dots,t\}$. Since the graph is acyclic, we then have that also the sequences $S(e_1),\dots,S(e_t)$ are pairwise unreachable. As such, they must appear in different \stpaths in any solution. Thus, for acyclic graphs, for each edge $e_i = uv \in A$ one can fix $x_{uv,i} = 1$, for all edges $uv$ in sequence $S(e_i)$. For a large number of $x$ variables to be set in this manner,~\cite{sena2025safe} found an antichain with the property that the sum of the lengths of the edges in $A$, i.e.~$\sum_{e \in A} w(e)$, is maximized. This problem of computing a \emph{maximum-weight antichain} can be solved in polynomial time with a reduction to a min-flow problem~\cite{rival2012graphs}.

The reasoning behind such maximum-weight antichain is that (i) safe sequences containing edges of such an antichain must appear in different solution \stpaths (call them \emph{incompatible}), as we also argue below for \emph{\stwalks}; and (ii) maximizing the total weight maximizes the number of path variables $x_{uv,i}$ that can be fixed to (at least) 1 in a preprocessing step.

To adapt this approach to general \stgraphs, possibly with cycles, we first show that we can analogously compute a maximum-weight antichain in a general graph with cycles, with the same complexity as for acyclic graphs (proof in \Cref{apx:missing-proofs}): 

\begin{theorem}
    \label{thm:antichain-reduction}
    Let $G = (V,E)$ be a graph, and let $w : E \rightarrow \mathbb{Z}_{\geq0}$. Computing a maximum-weight antichain of edges of $G$, namely a set $A = \{e_1,\dots,e_t\}$ of pairwise unreachable edges maximizing $\sum_{e \in A} w(e)$, can be reduced in time $O(|V| + |E|)$ to computing a maximum-weight antichain of edges of a DAG $G' = (V',E')$, with edge weights $w' : E' \rightarrow \mathbb{Z}_{\geq0}$, where $|V'| \leq 2|V|$ and $|E'| \leq |E|$.
\end{theorem}

\begin{figure}[t]
    \begin{subfigure}[m]{0.32\textwidth}
        \centering
        \includegraphics[width=\textwidth]{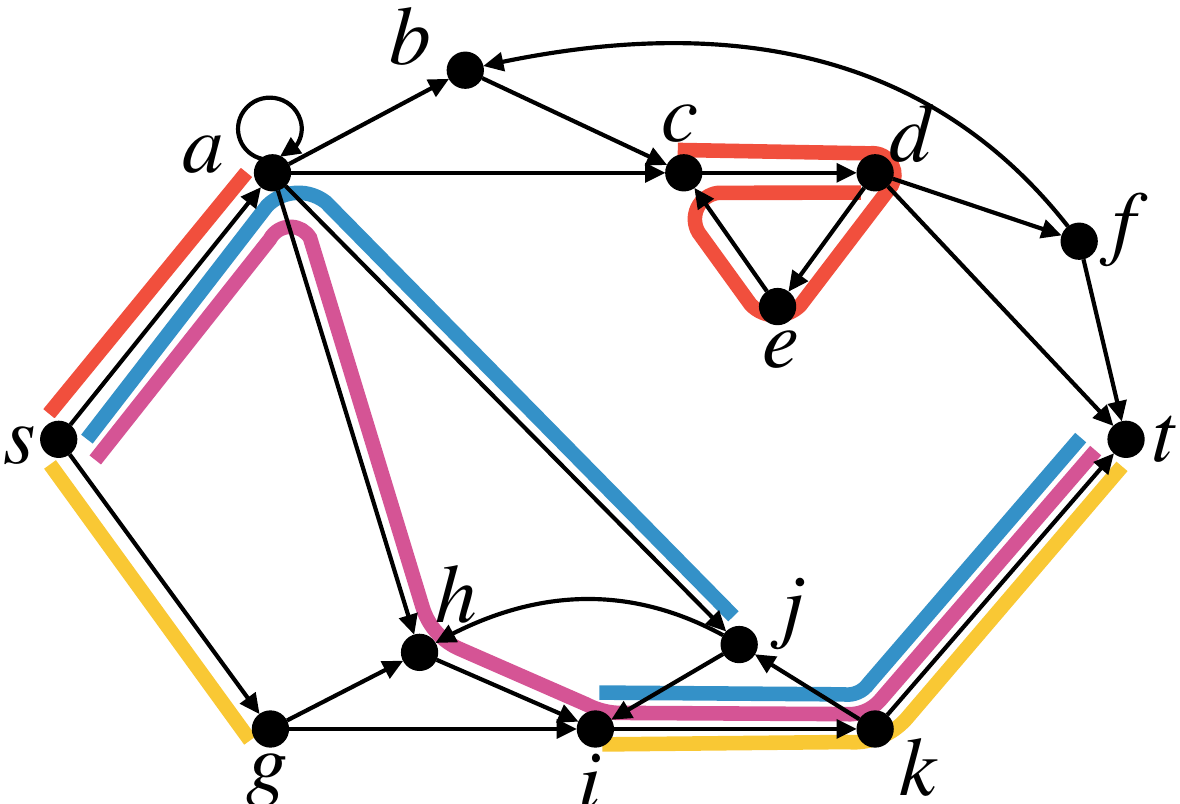}    
        \caption{Four incompatible safe sequences (of edges). They are drawn as contiguous lines for clarity, but they are in fact \textit{sequences} of edges.\label{fig:antichain-2}}
    \end{subfigure}
    \hfill
    \begin{subfigure}[m]{0.32\textwidth}
        \centering
        \includegraphics[width=\textwidth]{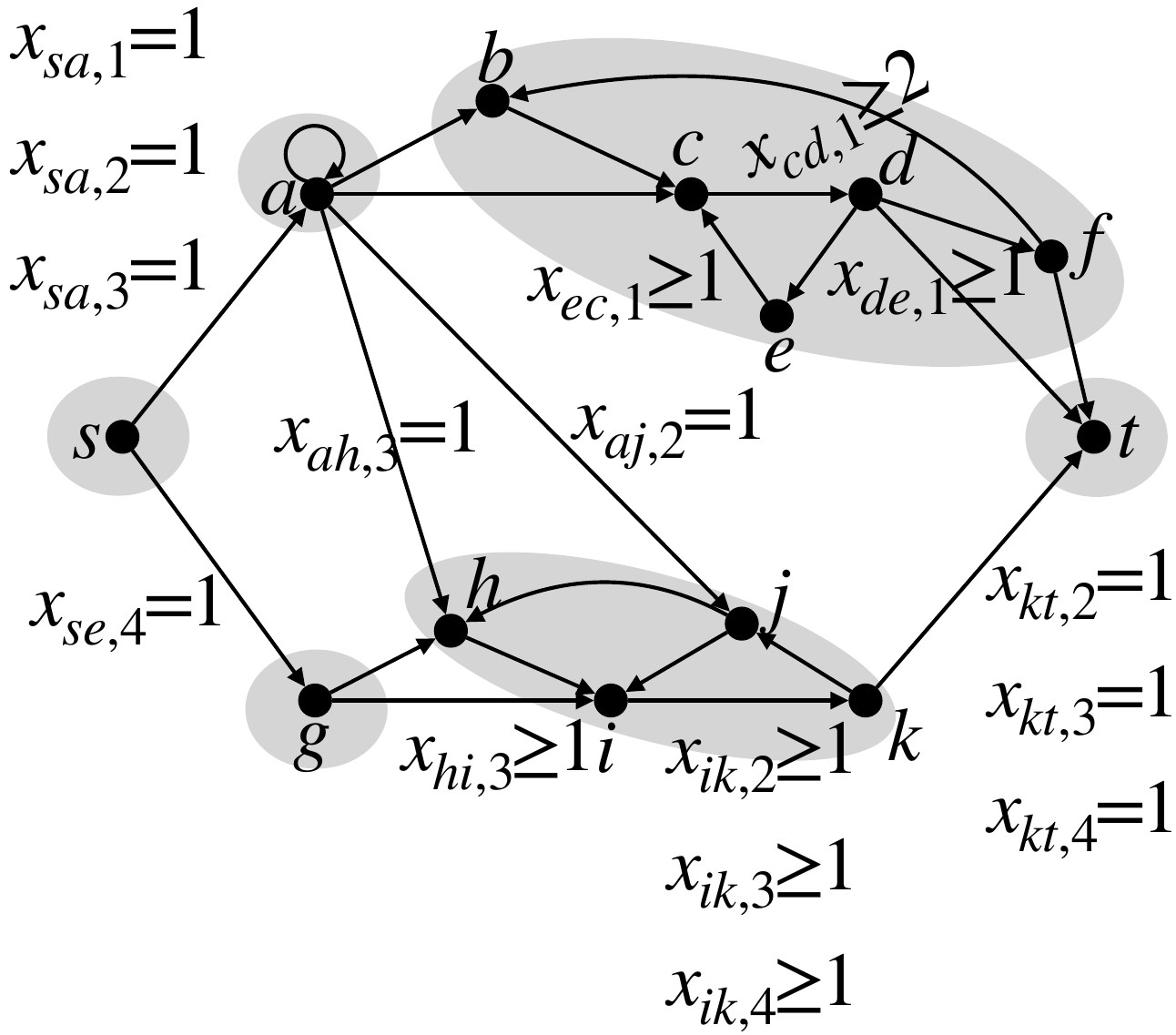}    
        \caption{Fixing walk variables to (at least) 1 using incompatible safe sequences. Strongly connected components are drawn in gray.\label{fig:antichain-fixing-1}}
    \end{subfigure}
    \hfill
    \begin{subfigure}[m]{0.32\textwidth}
        \centering
        \includegraphics[width=\textwidth]{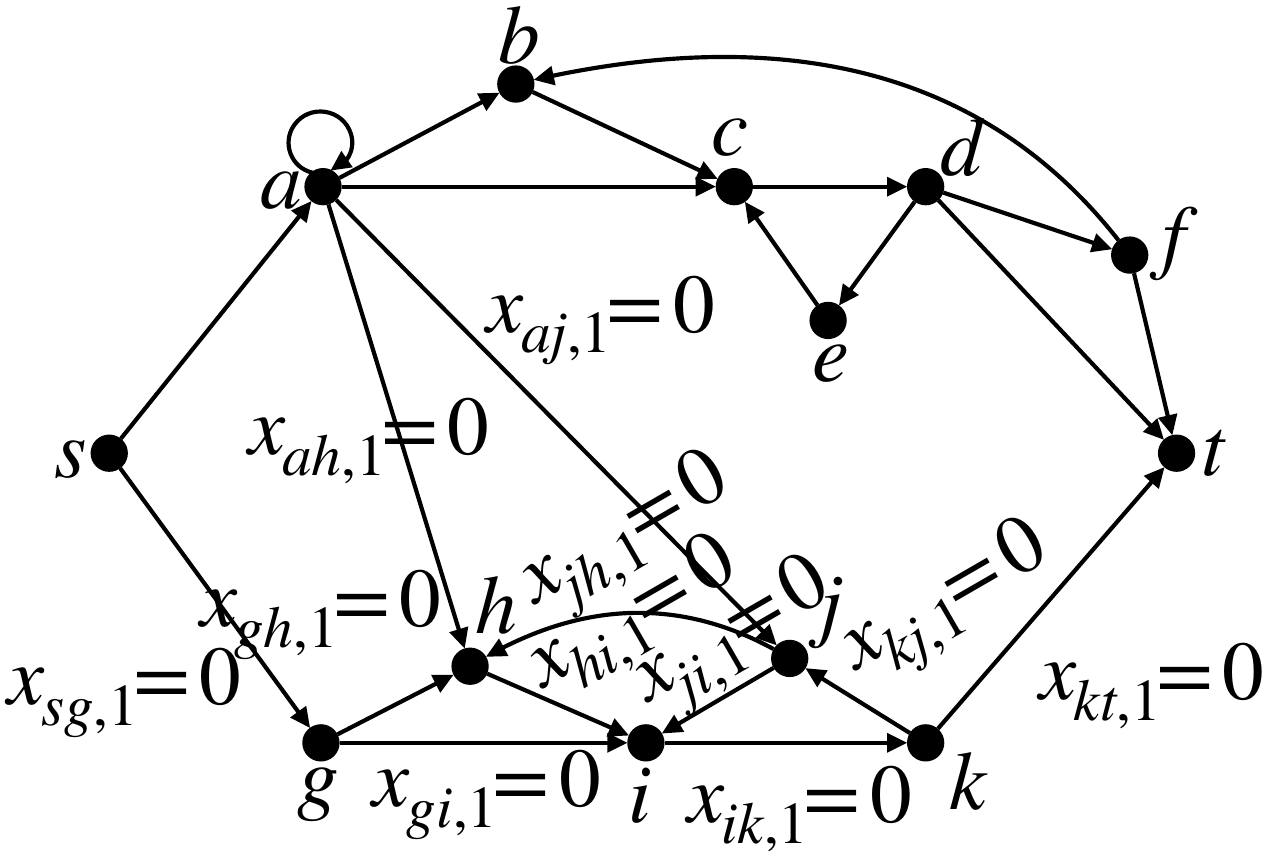}    
        \caption{Fixing walk variables to 0 using the red safe sequence.\label{fig:antichain-fixing-0}}
    \end{subfigure}

    \caption{Fixing walk variables using incompatible safe sequences of edges. For the $i$-th safe sequence (assumed here in the order red, blue, violet, yellow), and for every edge $uv$ in it, we set $x_{uv,i} = 1$ if $uv$ connects different SCCs (it cannot be traversed more than once by the $i$-th walk, since $G^{scc}$ is acyclic). If $uv$ lies within an SCC, we set $x_{uv,i}$ to be at least the number of times $uv$ appears in the sequence (it must be traversed at least this many times, but possibly more). When fixing variables to 0, note that for the red sequence (1st), edge $(j,h)$ does not reach $s$, is not reached by $t$, and while $j$ is reachable from $a$, $h$ does not reach $c$—thus it cannot be used by the first solution walk. 
    \label{fig:antichain-fixing}}
\end{figure}

Next, we show how to use such an antichain $A = \{e_1,\dots,e_t\}$ to fix $x_{uv,i}$ variables encoding \stwalks. Let $S(e_1),\dots,S(e_t)$ be sequences of edges passing through $e_1,\dots,e_t$. We claim that there can be no \stwalk containing two distinct sequences, say $S(e_i)$ and $S(e_j)$. Indeed, if this were so, then both $e_i$ and $e_j$ would appear on a same \stwalk, and thus either $e_i$ reaches $e_j$ or $e_j$ reaches $e_i$, which contradicts the fact that $A$ is an antichain. 

Therefore, we have that any solution to \Cref{def:problems} made up of \stwalk $W_1,\dots,W_k$ is such that: (i) each $S(e_i)$ belongs to a distinct \stwalk; without loss of generality, we can assume that $S(e_i)$ belongs to walk $W_i$; and (ii), as a consequence, $k \geq t$. In analogy to the walk notation, for a sequence $S(e_i)$, let $S(e_i)(uv)$ denote the number of times $uv$ appears in $S(e_i)$. Based on these observations, we can fix the following variables.

\paragraph{Fixing variables to (at least) 1.} We can constrain the walk variables $x_{uv,i}$ of the MILPs by setting (see also \Cref{fig:antichain-fixing-1}):
\vspace{-.2cm}
\begin{align}
    \forall i \in \{1,\dots,t\}, \forall uv \in S(e_i), \quad x_{uv,i} 
    \begin{cases}
    = 1, \text{if $uv$ is an edge between different SCCs of $G$,}\\
    \geq S(e_i)(uv), \text{if $uv$ is an edge inside some SCC of $G$.}\\
    \end{cases}
\end{align}
This is correct because edges between different SCCs cannot be traversed more than once by any \stwalk, while edges inside an SCC can be repeated by a walk. Moreover, we can use this observation also to transform \emph{any} integer variable $x_{uv,i}$ (for all $uv \in E$, and all $i \in \{1,\dots,k\}$) to a binary variable belonging to $\{0,1\}$, whenever $uv$ is an edge between different SCCs of the graph.

\paragraph{Fixing variables to 0.} Once we fix the edges of a safe sequence $S(e_i)$, we can further infer that some other edges $uv$ \emph{cannot} be part of the $i$-th solution walk, in which case we can set $x_{uv,i} = 0$. These are the edges such that \emph{none} of the following conditions hold (see \Cref{fig:antichain-fixing-0} for an illustration): 
\begin{enumerate}
\setlength\itemsep{0pt}
\setlength\parskip{0pt}
\setlength\parsep{0pt}
\item $v$ reaches the first vertex of the safe sequence, i.e.~the $i$-th walk could pass through $uv$ before passing through $S(e_i)$;
\item the last vertex of the safe sequence reaches $u$, i.e.~the $i$-th walk could pass through $uv$ after passing through $S(e_i)$;
\item for two consecutive edges $ab$ and $cd$ in the sequence, $b$ reaches $u$ and $v$ reaches $c$, i.e.~$i$-th walk could pass through $uv$ between two consecutive edges $ab$ and $cd$ in $S(e_i)$.
\end{enumerate}

\section{Experimental results}

\paragraph{Data.} In this paper we experiment with graphs created from four bacterial datasets: \ecoli, containing 30 E.coli strains from the dataset~\cite{ecoli}, \medium and \complex, containing 20 and 32 bacterial genomes, across 19 and 26 genera, respectively, from~\cite{shafranskaya2022metagt}, and \jgi, a mock community containing 23 bacterial strains and 3 archaeal strains, across 22 genera, from~\cite{singer2016next}. For the latter three, we use the genome abundances provided with the datasets. For \ecoli, where there are no abundances, we simulate them following the lognormal distribution with mean 1 and variance 1, multiplied by 10 and rounded up to the nearest integer.

To create datasets of increasing complexity, we create the graphs from an increasing number of genomes ($5, 10, 15, \dots$), until including all genomes in the dataset. Further, to create graphs of a feasible size for the ILP models, we partition these genomes into strings (windows) of a fixed length $\ell$. Because of the dataset complexity, in \medium and \jgi we use $\ell = 50,000$, in \complex we use $\ell = 10,000$, and in \ecoli we use $\ell = 5,000$. For the $i$-th window in each genome, we construct the $i$-th graph as an edge-centric de Bruijn graph, as follows. The nodes are $(k-1)$-mers, and the edges are $k$-mers, for a fixed value $k$. \rev{ 
High values of $k$ make the graphs less tangled, which in principle should accelerate the ILPs. On the other hand, low values of $k$ make the graphs more tangled and should slow down the ILPs. Thus, as middle ground and in line also with other studies such as~\cite{cracco2023extremely}, in all datasets we set $k = 63$.} \rev{Note that using full genomes (and not windows) results in graphs of tens of thousands of nodes (e.g. 45,196 nodes and 61,182 edges on five genomes from \complex), which is currently fully impractical with an ILP solver.} For each occurrence of a $k$-mer in a genome, we increase the flow value of the corresponding edge by the abundance value of the genome. The genome windows then correspond to walks in this graph, and moreover they, with the abundances used to create the graph, form a flow decomposition. In this graph, we compact all unitigs (i.e.~paths whose internal nodes have in-degree and out-degree equal to 1). Finally, we keep only those graphs that contain at least one cycle. 

The above creates perfect edge weights where flow conservation holds at every node. For imperfect data, we follow~\cite{dias2024robust} and replace each error-free edge weight $x$ by a sample from the Poisson distribution $Pois(x)$. To provide subset constraints to the models, when creating the $i$-th graph, we also simulate 5 reads from each genome window, from a location chosen uniformly at random, by reading (at most) 1,000 bases until the end of the genome window. This string is then mapped to a walk in the de Bruijn graph, by traversing each edge in the order of its $k$-mers. All graphs are available at \url{https://doi.org/10.5281/zenodo.17549958}.

\paragraph{Implementation.} We implemented the ILPs for the three problems from \Cref{def:problems}, and described in \Cref{sec:milp-formulations}, in the \flowpaths Python package, see \url{https://algbio.github.io/flowpaths}. This package uses by default the popular \highs MILP solver, see \url{www.highs.dev} and~\cite{huangfu2018parallelizing}, ensuring that our implementations can be immediately used without procuring a license. We used the \highs Python API (v1.11.0) that gets installed automatically with \flowpaths. \rev{Note that \highs currently supports only single-threaded MILP solving, but a multi-threaded version is expected in Summer 2026.\footnote{\url{https://ergo-code.github.io/HiGHS/dev/solvers/}}}
See~\Cref{sec:flowpaths} for minimal code examples of these models. We ran our experiments with version v0.2.12 of \flowpaths, on a machine with an AMD \rev{ThreadRipper PRO 3955WX} processor and 512GB RAM, and with a timeout of 300 seconds for \highs. 

To fix a value of $k$ in~\Cref{def:problems}, for $k$-FD we use the smallest value for which $k$-FD exists, i.e.~we solve the minimum flow decomposition problem (we call the resulting model \ILPMFD). For $k$-LAE and $k$-MPE, $k$ is the minimum number of $s$-$t$ walks needed to cover the edges in $E'$, computed with \Cref{thm:antichain-reduction}. 

The implementations for \ILPMFD, \ILPLAE and \ILPRobust support setting the set of edges $E'$ from~\Cref{def:problems}, support arbitrary sets $S$ and $T$ of start and end nodes for the solution walks (recall the discussion after \Cref{def:problems}), and support subset constraints. For the \ILPMFD and \ILPRobust models, any solution is also a walk cover of the edges in $E'$, by~\Cref{def:problems}. Hence, for them we use safe sequences with respect to walk covers of $C = E'$. For \ILPLAE, we cannot easily guarantee which edges will be traversed by the solution walks. For it, we add to $E'$ only those edges in the subset constraints, since by definition, they must be covered by the solution walks, thus guaranteeing optimality.

\rev{To test different experimental scenarios, for \ILPMFD we use $E'$ as the set of edges with non-zero weight and no subset constraints, for \ILPLAE we use $E' = E$ and the subset constraints from our simulation (and use $C$-safe sequence with $C$ as the edges of the subset constraints). For \ILPRobust, note that all edges with non-zero weight must be covered by some solution walk. However, in practical applications there might be noise in the weights, and for example the edges with low weight might be erroneous and not be part of any solution. To evaluate the behavior of the ILPs in this scenario, we set $E'$ as the set of those edges whose flow value is above the 25-th percentile of all edge values, and no subset constraints. By definition, all edges in this set $E'$ must appear in a solution walk (note that also other edges in $E$ can be used by the solution walks), thus $C$-safe sequences for $C = E'$ still guarantee optimality.}

\begin{table}[t]
\caption{\textbf{Experimental results on perfect data with \ILPMFD.} Column ``\#gen'' contains the number of genomes used to create the graphs, ``\#graphs'' contains the number of graphs (with cycles) obtained from all windows of that number of genomes, ``avg $n$ / $m$'' contains the average number of nodes / edges of the graphs, while in parentheses we list the maximum number of nodes / edges of the graphs. Column ``prep (s)'' contains the average time taken by the safety optimizations described in \Cref{sec:simplifying-milp}, ``\#solved'' contains the number of solved instances within the 300s timeout, ``time (s)'' contains the average time taken by the solved instances (without and with the safety optimizations, from beginning to end), and ``speed-up $(\times$)'' contains the average speed-up obtained with the safety optimizations, where if an instance timed out, we considered it to run in 300s.\label{tab:mfd}}
\centering
\scalebox{0.85}{
\begin{tabular}{|r|r|r|r|r|r|r|r|r|}
\hline
& \multirow{2}{*}{\#gen} 
& \multirow{2}{*}{\#graphs} 
& \multirow{2}{*}{\shortstack{avg $n$\\(max $n$)}} 
& \multirow{2}{*}{\shortstack{avg $m$\\(max $m$)}} 
& \multirow{2}{*}{prep (s)} 
& \multicolumn{2}{c|}{\#solved, time (s)} 
& \multirow{2}{*}{speed-up $(\times$)} \\ \cline{7-8}
& & & & & & no safety & safety & \\ \hline

\multirow{7}{*}{\rotatebox{90}{\footnotesize\textbf{complex32}}}
& 5 & 63 & 26 (144) & 37 (194) & 0.009 & 63, 1.00 & 63, 0.08 & 13.8 \\
& 10 & 162 & 28 (151) & 44 (209) & 0.015 & 162, 4.34 & 162, 0.14 & 28.5 \\
& 15 & 114 & 33 (156) & 54 (219) & 0.024 & 111, 9.13 & 114, 0.20 & 58.7 \\
& 20 & 116 & 42 (226) & 69 (318) & 0.041 & 113, 18.61 & 116, 0.30 & 65.3 \\
& 25 & 118 & 70 (284) & 111 (406) & 0.087 & 106, 42.04 & 118, 0.66 & 94.2 \\
& 30 & 122 & 77 (289) & 124 (416) & 0.106 & 97, 50.26 & 122, 0.83 & 111.3 \\
& 32 & 124 & 80 (291) & 131 (420) & 0.115 & 99, 56.74 & 124, 0.87 & 115.4 \\

\hline

\multirow{6}{*}{\rotatebox{90}{\footnotesize\textbf{ecoli}}}
& 5 & 127 & 11 (108) & 17 (157) & 0.004 & 127, 0.49 & 127, 0.04 & 13.1 \\
& 10 & 209 & 30 (249) & 46 (370) & 0.016 & 208, 7.84 & 209, 0.20 & 55.0 \\
& 15 & 271 & 53 (334) & 82 (502) & 0.036 & 255, 24.97 & 270, 0.45 & 107.5 \\
& 20 & 369 & 84 (398) & 128 (614) & 0.072 & 311, 45.91 & 367, 0.87 & 138.7 \\
& 25 & 436 & 99 (438) & 151 (672) & 0.107 & 318, 54.53 & 432, 2.26 & 143.8 \\
& 30 & 487 & 139 (466) & 210 (721) & 0.194 & 288, 78.27 & 476, 1.99 & 128.6 \\

\hline

\multirow{5}{*}{\rotatebox{90}{\footnotesize\textbf{JGI}}}
& 5 & 47 & 16 (86) & 25 (127) & 0.005 & 47, 1.93 & 47, 0.05 & 20.7 \\
& 10 & 41 & 34 (122) & 52 (175) & 0.019 & 37, 16.53 & 41, 0.17 & 165.8 \\
& 15 & 44 & 41 (132) & 66 (194) & 0.032 & 37, 44.84 & 44, 0.25 & 306.3 \\
& 20 & 42 & 53 (151) & 87 (224) & 0.048 & 30, 76.25 & 42, 0.40 & 345.7 \\
& 26 & 38 & 67 (181) & 111 (272) & 0.081 & 22, 88.11 & 38, 0.68 & 315.8 \\

\hline

\multirow{4}{*}{\rotatebox{90}{\footnotesize\textbf{medium20}}}
& 5 & 52 & 16 (75) & 24 (108) & 0.005 & 52, 2.06 & 52, 0.04 & 27.5 \\
& 10 & 39 & 26 (86) & 42 (131) & 0.015 & 39, 15.52 & 39, 0.11 & 96.6 \\
& 15 & 40 & 36 (113) & 60 (170) & 0.026 & 36, 37.60 & 40, 0.22 & 216.1 \\
& 20 & 40 & 50 (125) & 84 (200) & 0.042 & 28, 76.88 & 40, 0.34 & 393.2 \\

\hline
\end{tabular}}
\end{table}

\paragraph{Discussion.} The experimental results for \ILPMFD are in \Cref{tab:mfd}. The results for \ILPLAE and \ILPRobust are in \Cref{tab:lae,tab:mpe} in \Cref{sec:additional-experiments}. We generally observe that, without optimizations, the running times (and the proportion of instances where the MILP solver times out) increase with the complexity of the MILP models (in the order \ILPMFD, \ILPLAE, \ILPRobust), and with the number of genomes used to create the graphs (column ``\#gen'').

Overall, we observe that the time taken to preprocess the MILP as described in~\Cref{sec:simplifying-milp} is negligible (often less than 0.1s on average), confirming efficiency. Applying these optimizations significantly increases the number of instances solved within the time limit, and presents significant speed-ups. 

For \ILPMFD all instances are solved, except for a small number in the \ecoli dataset. The average running time of the optimized MILPs from beginning to end is generally under 1 second on average, presenting speed-ups of up to 393$\times$ (on \medium, 20 genomes). For \ILPLAE, note that we also add subset constraints, and the edges $E'$ used for safety are only from the subset constraints, and not all edges of the graph. The number of solved instances is less than for \ILPMFD, especially when more genomes are used to construct the graphs. We observe speed-ups of up to 560$\times$ (on \complex, 15 genomes); the speed-up values decrease for graphs constructed from more genomes because the number of instances timing out increases. For \ILPRobust, recall that $E'$ is computed by selecting edges whose flow value is larger than the 25-th percentile. This means that safe sequences are smaller than if setting $E' = E$ (because less edges are guaranteed to appear in a solution walk), but we still observe speedups of more than 1000$\times$ for graphs created from 10 genomes. Overall, our optimized \ILPRobust solves most instances in under one minute on average, and solves most instances within the timeout (except for \ecoli graphs). We also observe more solved instances compared to \ILPLAE, because of the different set $E'$ used to compute safety.

\section{Conclusions}

In this work we generalize the methodology of~\cite{sena2025safe} from DAGs to general graphs, i.e., graphs that contain cycles. To do so, we develop a theory on dominators, relate it to safety on walk-covering problems, and propose the notion of univocal path in the dominator trees. Further, we show that collapsing univocal paths in the dominator trees preserves the dominance relation and yields a natural generalization of the concept of a unitig (in \stgraphs). We also introduce new ILP safety-preprocessing techniques, namely considering the condensation DAG to set variables to at-least-one or to exactly-one, and setting variables to 0 between mutually unreachable safe sequences. Notice that our framework applies to any assembly-like problem whose space of solutions consists of a walk cover of the vertices or edges of the graph, or subsets thereof.

Together, our scaling techniques substantially speed-up the ILP solver in three existing decomposition models and enable many instances to be solved that would otherwise time out. However, such an exact ILP approach at the moment works only for graphs with hundreds of vertices, which is why we had to partition the genomes into windows of suitable lengths. Moreover, even for the graphs constructed in this windowed manner, the experiments show limitations of the framework. For example, on 30 E. coli genomes (where the largest graph has 466 vertices and 721 edges), \ILPLAE already times out for more than half of the 487 graphs (see~\Cref{tab:lae}). Similarly, also \ILPRobust times out for 97 out of 487 graphs (see~\Cref{tab:mpe}).

Overall, this work serves as a stepping stone for solving exact formulations of multi-assembly linear programs that accommodate features observed in real-world data. As it has been done for acyclic graphs, it would be interesting to assess the accuracy of the linear programs considered in this work on general graphs, something that was not possible before due to the intrinsic solving difficulty of these programs.
Our framework allows practitioners to test and fine-tune these models in larger graphs. Our methods are \emph{sound}, in the sense that the space of solutions of the ILPs does not change if given appropriate safety information.

\newpage

\bibliography{refs}

\newpage

\appendix 

\section{MILP formulations for flow decomposition problems into set of weighted walks}
\label{sec:milp-formulations}

In this section we describe the MILP models for the three problems defined in \Cref{def:problems}. These models use the same main ideas as the MILP model introduced in~\cite{dias2022minimum} for the $k$-FD problem in cyclic graphs, but improve upon it by having fewer constraints. The formulation follows the same same general methodology as for flow decomposition in acyclic graphs: first formulate $k$ \stwalks (now walks instead of paths)---as in the first part of \Cref{def:problems}; second, add constraints for each specific problem objective, as in the second part of~\Cref{def:problems}. The main difficulty here with respect to acyclic graphs is to model \stwalks via MILP.

\subsection{Formulation of \stwalks}
\label{sec:formulation-of-walks}

We first introduce suitable variables and constraints that model $k$ \stwalks in an \stgraph. As in~\cite{dias2022minimum}, for each walk $W_i$, and every edge $uv \in E$, we introduce an integer variable $x_{uv,i} \in \mathbb{Z}_{\geq 0}$ representing the number of times the edge $uv$ is traversed by the walk $W_i$. To ensure that the variables $x_{uv,i}$ correctly model a walk, we use the following lemma from \cite{dias2022minimum}. 

\begin{lemma}[\cite{dias2022minimum}]
\label{lem:walk}
Let $G = (V,E)$ be an \stgraph, and let $W$ be a multiset of edges of $G$, where for every edge $uv \in W$, we denote by $W(uv)$ the number of times the edge appears in the multiset $W$. It holds that the edges of $W$ can be ordered to form an \stwalk passing $W(uv)$ times through each edge $uv \in E$ if and only if the following conditions hold:
\begin{enumerate}
\item For every $v \in V$, $\sum_{uv \in E} W(uv) - \sum_{vw \in E} W(vw) = 
\begin{cases}
-1 & \text{if } v = s, \\
1 & \text{if } v = t, \\
0 & \text{otherwise}.
\end{cases}$
\item For every vertex $v$ appearing in some edge of the multiset $W$, there is an $s$-$v$ path using only edges in $W$ (i.e. $v$ is reachable from $s$ using edges in $W$).
\end{enumerate}
\end{lemma}

Condition 1 of Lemma \ref{lem:walk} can be directly modeled with the following constraints, for each vertex $v \in V$:
\begin{align}
\sum_{uv \in E} x_{uv,i} - \sum_{vw \in E} x_{vw,i} = 
\begin{cases}
-1 & \text{if } v = s, \\
1 & \text{if } v = t, \\
0 & \text{otherwise},
\end{cases} \quad \forall v \in V \setminus \{s,t\}, \forall i \in \{1,\dots,k\}.\label{eq:walk-1}
\end{align}

To model condition 2 of Lemma \ref{lem:walk}, the main idea from~\cite{dias2022minimum} is to model a \emph{reachability tree} from $s$ to every vertex $v$ appearing in some edge $uv$ such that $x_{uv,i} > 0$, and made up only of edges with positive $x$ values. That is, we want to ensure that if $x_{uv,i} > 0$, then there is a path from $s$ to $v$ using only edges in the tree. 

In this paper we offer a simpler formulation of such a reachability tree than in~\cite{dias2022minimum}, as follows. First, as in~\cite{dias2022minimum}, we introduce binary variables $y_{uv,i} \in \{0,1\}$ with the meaning that $y_{uv,i} = 1$ iff the edge $uv$ is part of this tree. We first state that an edge $uv$ can be in the tree (i.e. $y_{uv,i} = 1$) only if $x_{uv,i} \geq 1$:
\begin{align}
y_{uv,i} &\leq x_{uv,i} \quad \forall uv \in E, \forall i \in \{1,\dots,k\}.\label{eq:walk-2}
\end{align}
Next, we state that if a vertex $v$ appears in some edge $uv$ such that $x_{uv,i} \geq 1$ (call such a vertex \emph{selected}), then there is one, and exactly one, incoming edge $wv$ such that $y_{wv,i} = 1$ (call such an edge \emph{selected}):
\begin{align}
    \sum_{wv \in E} x_{wv,i} &\leq M_1 \sum_{wv \in E} y_{wv,i} \quad \forall v \in V, \forall i \in \{1,\dots,k\},\label{eq:walk-3}\\
    \sum_{wv \in E} y_{wv,i} &= 1 \quad \forall v \in V, \forall i \in \{1,\dots,k\},\label{eq:walk-4}
\end{align}
where $M_1$ is a sufficiently large upper-bound for the sum on the left-hand side of \eqref{eq:walk-3}, for example the in-degree of $v$, which we denote $d^-(v)$, multiplied by the maximum number of times that a walk in the solution can visit an edge. 

At this point, we have that for every selected vertex $v$, there is exactly one incoming selected edge $wv$ (i.e.~such that $y_{wv,i} = 1$). This means that such edges induce a collection of disjoint cycles, with trees ``hanging'' from them, and whose edges are oriented downwards. To ensure that they induce a \emph{single tree}, rooted at $s$, we introduce a \emph{distance-like} variable $d_{v,i} \in \mathbb{Z}_{\geq 0}$ for each vertex $v$ and every walk $i$. We set $d_{s,i} = 0$ and for every other vertex $v$, we have the following constraint:
\begin{align}
d_{v,i} &\geq d_{u,i} + 1 - M_2 (1-y_{uv,i}) \quad \forall uv \in E, \forall i \in \{1,\dots,k\},\label{eq:walk-5}
\end{align}
where $M_2$ is a sufficiently large upper-bound for $d_{v,i}$, e.g., the number of vertices in the graph. The interpretation of this constraint is as follows. If $y_{uv,i} = 0$ (i.e.~the edge $uv$ is not in the tree), then the constraint is vacuous since the right-hand side is is non-positive. However, if $y_{uv,i} = 1$, then the constraint becomes $d_{v,i} \geq d_{u,i} + 1$, i.e. the distance of $v$ is at least one more than that of $u$. This ensures that there are no cycles induced by the edges with $y_{uv,i} = 1$, since otherwise they would induce a cycle of unsatisfiable inequalities. This ensures that they induce a forest, but because of \cref{eq:walk-1}, then they induce a single tree, rooted at $s$. 

The key difference between our model for \stwalks and that of~\cite{dias2022minimum} lies in how the properties enforced by~\cref{eq:walk-5} are modeled. In~\cite{dias2022minimum}, these properties were expressed through a product of the $y$ and $d$ variables. Although such products can be linearized (as we will discuss in \Cref{sec:formulation-problem-objectives}), doing so introduces extra variables and constraints, which in turn slows down the MILP solver. In contrast, our approach avoids this complication, since~\cref{eq:walk-5} is already a linear constraint.

We summarize the correctness of our model in the following theorem, whose proof is in \Cref{apx:missing-proofs}.

\begin{theorem}
    \label{thm:walk-correctness}
    Let $G = (V,E)$ be an \stgraph, and let $x_{uv,i} \in \mathbb{Z}_{\geq 0}$ be variables satisfying \cref{eq:walk-1}. It holds that there exists an $s$-$t$ walk $W_i$ such that $W_i(uv) = x_{uv,i}$ for every edge $uv \in E$, if and only if the MILP model made up of \cref{eq:walk-2,eq:walk-3,eq:walk-4,eq:walk-5} is feasible.
\end{theorem}
\begin{proof}
    ($\Rightarrow$) Suppose that there exists an $s$-$t$ walk $W_i$ such that $W_i(uv) = x_{uv,i}$ for every edge $uv \in E$. Consider the subgraph $G'$ of $G$ consisting only of edges that are present in $W$, and consider a breath-first search tree $T$ in this graph, rooted as $s$. For every edge $uv \in E$, set $y_{uv,i} = 1$ if $uv$ is also an edge of $T$, and $y_{uv,i} = 0$ otherwise. Moreover, for every vertex $v \in V$, set $d_{v,i} = 0$ if $v = s$ or $v$ is not in $G'$, and otherwise set $d_{v,i}$ as the length of a shortest path from $s$ to $v$ in $G'$ (in terms of number of edges). We claim that these assignments of $y_{uv,i}$ and $d_{v,i}$ variables satisfy \cref{eq:walk-2,eq:walk-3,eq:walk-4,eq:walk-5}:
    \begin{itemize}
        \item \Cref{eq:walk-2} is satisfied because $y_{uv,i} = 1$ only if $uv$ is also an edge of $T$, which implies $x_{uv,i} \geq 1$.
        \item For \Cref{eq:walk-3}, if the sum $\sum_{wv \in E} x_{wv,i}$ is zero, then the condition trivially holds. Otherwise, observe that $\sum_{wv \in E} x_{wv,i} \leq d^-(v) \max_{uv\in E} x_{uv,i}$ and we choose $M_1 \geq d^-(v) \max_{uv\in E} x_{uv,i}$. Since $\sum_{wv \in E} x_{wv,i}$ is not zero, then $v$ is a vertex of $G'$ and thus it has an incoming edge in the breath-first search tree $T$, say $uv$. Since we set $y_{uv,i} = 1$ for the edges of the tree, we have that $\sum_{wv \in E} y_{wv,i} \geq 1$ and thus \Cref{eq:walk-3} holds.
        \item \Cref{eq:walk-4} holds by the fact that $T$ is a tree, and for every vertex $v$ in $G'$ (except $s$) we set $y_{uv,i} = 1$ for exactly one edge $uv$ in-coming to $v$ (i.e. the edge of the tree that is in-coming to $v$).
        \item For \Cref{eq:walk-5}, note that if $y_{uv,i} = 0$, the condition holds, as $M_2$ is chosen larger than the number of the vertices of the graph, and the values of $d_{v,i}$ are distances from $s$, which are at most the number of vertices of the graph. If $y_{uv,i} = 1$, since then $uv$ is an edge on a shortest path from $s$ to $v$, and thus the distance $d_{v,i}$ from $s$ to $v$ is equal the distance $d_{u,i}$ from $s$ to $u$, plus one, and thus \Cref{eq:walk-5} holds.
    \end{itemize}
    ($\Leftarrow$) Suppose that the MILP model made up of \cref{eq:walk-2,eq:walk-3,eq:walk-4,eq:walk-5} is feasible. We need to show that there exists an $s$-$t$ walk $W_i$ such that $W_i(uv) = x_{uv,i}$ for every edge $uv \in E$. Thanks to \Cref{lem:walk}, it suffices to show that for every selected vertex $v$ (i.e.~for which $x_{uv,i} \geq 1$ for some $uv \in E$), there is an $s$-$v$ path using only edges $zw$ having positive $x_{zw,i}$ value. 

    Consider the selected edges $uv$ (i.e.~for which $y_{uv,i} = 1$). By \cref{eq:walk-3,eq:walk-4}, we have that every selected vertex $v$ has exactly one in-coming selected edge. Let $G'$ be the subgraph of $G$ made up of selected edges; we claim that $G'$ has no cycles. Indeed, if there were a cycle made up of selected edges $v_1v_2$, $v_2v_3$, \dots, $v_{t-1}v_{t}$, $v_{t}v_{1}$, then \cref{eq:walk-5} applied to them would imply $v_2 > v_1$, $v_3 > v_2$, \dots, $v_{t} > v_{t-1}$, $v_1 > v_t$, thus $v_1 > v_1$, a contradiction. Thus, $G'$ is a forest, made up of trees directed away from their roots.

    Let $s'$ be an arbitrary root in some tree in $G'$ and note that each tree in the forest contains at least one edge (recall that we created $G'$ from a set of edges). Thus, there is some edge $s'v$ such that $y_{s'v,i} = 1$. By \cref{eq:walk-2}, it follows that $x_{s'v,i} \geq y_{s'v,i} \geq 1$. Since the $x$ variables satisfy \cref{eq:walk-1}, if $s' \neq s$, then there is at least one in-coming edge of $s'$, say $us'$, such that $x_{us'} \geq 1$. Therefore, $s'$ is a selected vertex and therefore $y_{u's'} = 1$ for some in-coming edge $u's'$ of $s'$ (possibly $u' = u$), by \cref{eq:walk-3,eq:walk-4}. Therefore, also the edge $u's'$ is in $G'$, which contradicts the fact that $s'$ is the root of a tree of $G'$. 
    
    Therefore, it follows that $s' = s$, and thus $G'$ is made up of a single tree rooted at $s$, whose edges are directed away from $s$, and which contains all selected vertices (i.e.~those vertices $v$ for which $x_{uv,i} \geq 1$ for some $uv \in E$). Thus, for every selected vertex $v$  we can obtain an $s$-$v$ path as the path in $G'$ from $s$ to $v$. All edges $uw$ on such a path have $y_{uw,i} = 1$, by definition of $G'$, and thus also $x_{uw,i} \geq 1$ (by \cref{eq:walk-2}).
\end{proof}

\subsection{Formulation of problem objectives}
\label{sec:formulation-problem-objectives}

Having a formulation for $k$ \stwalks, encoded as variables $x_{uv,i} \in \mathbb{Z}_{\geq 0}$, for all $uv \in E$, and $i \in \{1,\dots,k\}$, we now show how also the problem objectives from \Cref{def:problems} can be modeled in MILP. 

For each walk $i \in \{1,\dots,k\}$ we introduce a weight variable $w_i \in \mathbb{Z}_{\geq 0}$  (and also a slack variable $\rho_i$ for $k$-MPE). The problem objectives then require modeling the product $x_{uv,i}w_i$ (and also $x_{uv,i}\rho_i$ for $k$-MPE). Even though these are not linear terms, as required by a linear program, they can be linearized using a standard power-of-two technique technique, see~\cite[Remark 7]{dias2022minimum}. For completeness, we present it also here. Suppose we have a product between an integer variable $x$ and an arbitrary variable $y$, and that $\overline{x}$ is an upper-bound for $x$, and let $t = \lfloor\log_2(\overline{x})\rfloor$. Then we can introduce $t+1$ binary variables $x_0,\dots,x_t$ together with the linear constraint $x = \sum_{j=0}^{t}2^jx_j$. Then, we have that $xy = \sum_{j=0}^{t}2^jx_jy$. Each of $x_jy$ is a product between a binary variable $x_j$ and the variable $y$, which can in turn be linearized in a standard manner, as follows. Assuming that $y$ can take values only in the interval $[\underline{y},\overline{y}]$ we introduce a new variable $z_j$ to model this product, together with the constraints:
\begin{align}
    x_j\underline{y} \leq z_j &\leq  x_j\overline{y},\\
    y - \overline{y} (1 - x_j) \leq z_j &\leq y - \underline{y} (1 - x_j).
\end{align}

As such, the objective for $k$-FD is (assuming we further linearize non-linear terms as described above):
\begin{align}
    \sum_{i=1}^{k}x_{uv,i}w_i &= f(uv), \quad \forall uv \in E, \forall i \in \{1,\dots,k\}.
\end{align}
For $k$-LAE we introduce a new error variable $z_{uv} \in \mathbb{Z}_{\geq 0}$ for every edge $uv \in E'$, together with the constraints:
\begin{align}
    - z_{uv} \leq f(uv) - \sum_{i=1}^{k}x_{uv,i}w_i &\leq z_{uv}, \quad \forall uv \in E'.
\end{align}
Then, we require from the MILP to minimize the sum of errors, namely $\displaystyle\min \sum_{uv \in E'} z_{uv}$.

For $k$-MPE (assuming again we handle non-linear terms as above) it suffices to add the constraints:
\begin{align}
    - \sum_{i=1}^{k}x_{uv,i}\rho_i \leq f(uv) - \sum_{i=1}^{k}x_{uv,i}w_i \leq \sum_{i=1}^{k}x_{uv,i}\rho_i, \quad \forall uv \in E',
\end{align}
and require minimizing the sum of slacks $\displaystyle\min \sum_{i=1}^{k} \rho_i$.

\subsection{Formulation of subset constraints}

For general graphs with cycles, we introduce here a weaker form of a set of subpath constraints, but one which can be added to the MILP models developed in the previous subsections, i.e.~can be expressed as further constraints on the $x_{uv,i}$ variables modeling walks.

More specifically, assume that we are also given a family $\mathcal{S} = \{S_1,\dots,S_\ell\}$ of subsets of edges, namely $S_j \subseteq E$, for all $j \in \{1,\dots,\ell\}$. The decomposition problems defined in \Cref{def:problems} are then generalized to additionally require that for every set $S_j$, there is at least one walk $W_i$ such that $W_i$ contains all edges of $S_j$. Formally: \[\forall j \in \{1,\dots,\ell\}, \exists i \in \{1,\dots,k\}, \text{ such that } W_i(uv) \geq 1,\forall uv \in S_j.\]

Note that these constrains just apply to the walks, not to the objectives of the decomposition problems. As such, it suffices to model them on the $x_{uv,i}$ variables and they will be valid for all problems.

We proceed as in \cite{dias2022fast}, with the difference that we now need to handle solution walks and subset constraints (instead of solution paths and subpath constraints). First, we introduce binary variables $p_{uv,i} \in \{0,1\}$ which will indicate whether edge $uv \in E$ is present in walk $i \in \{1,\dots,k\}$, namely $p_{uv,i} = \min(1,x_{uv,i})$. We enforce this equality with the following two constraints:
\begin{align}
    p_{uv,i} &\leq x_{uv,i}, \quad \forall uv \in E, \forall i \in \{1,\dots,k\},\\
    x_{uv,i} &\leq M_3 p_{uv,i}, \quad \forall uv \in E, \forall i \in \{1,\dots,k\},
\end{align}
where $M_3$ is a suitably large constant, for example the maximum number of times an edge can appear in a walk.

For every set $S_j$ and every $i \in \{1,\dots,k\}$, we then introduce a binary indicator variable $s_{i,j} \in \{0,1\}$ which is allowed to equal 1 only if the walk $W_i$ contains all the edges of set $S_j$. To enforce this, we add the constraint:
\begin{align}
    \sum_{uv \in S_j} p_{uv,i} \geq |S_j|s_{i,j}, \quad \forall i \in \{1,\dots,k\}, \forall j \in \{1,\dots,\ell\}.
\end{align}
Finally, we require that every set $S_j$ is present in at least one solution walk $W_i$:
\begin{align}
    \sum_{i=1}^{k} s_{i,j} \geq 1, \quad \forall j \in \{1,\dots,\ell\}.
\end{align}

\section{Missing proofs}
\label{apx:missing-proofs}

This section contains the necessary proofs missing from the main text.

\begin{proof}[Proof of \Cref{lem:st-dom-immediate-relation}]
    Note that $u \neq t$ since $u$ is an $s$-dominator. Let $w:=dom_t(u,k)$ and let $y_1,\dots,y_{k-1}$ be the vertices between $u$ and $v$ in the $s$-dominator tree. We can assume $w \neq v$.

    Observe that $w \neq y_i$ for every $i \in \{1,\dots,k-1\}$. Suppose otherwise. Then $w \neq t$ because $t$ is not an $s$-dominator, and so there are $k-1$ vertices $x_1,\dots,x_{k-1}$ between $w$ and $u$ in the $t$-dominator tree. Some $x_j$ is not between $u$ and $v$ in the $s$-dominator tree, hence there is a $u$-$w$ path avoiding $x_j$. Further, since $x_j$ is strictly $t$-dominated by $w$ there is a $w$-$t$ path avoiding $x_j$. Concatenating these paths at $w$ results in a $u$-$t$ path avoiding $x_j$, a contradiction.

    Suppose now for a contradiction that $v$ is not strictly $t$-dominated by $w$. From the point above we know that there is a $u$-$v$ path avoiding $w$. This path concatenated with a $v$-$t$ path avoiding $w$ gives a $u$-$t$ path also avoiding $w$, contradicting the fact that $w$ $t$-dominates $u$.
    
    It remains to show that $v$ does not $t$-dominate $u$. Suppose otherwise. Since $w$ strictly $t$-dominates $v$ from the point above we have $v=x_i$ for some $i \in \{1,\dots,k-1\}$.
    Thus some $y_j$ is not between $v$ and $u$ in the $t$-dominator tree, implying that there is a $u$-$v$ path avoiding $y_j$. Since $y_j$ is $s$-dominated by $u$ there is an $s$-$u$ path avoiding $y_j$ which when concatenated with the previous path at $u$ gives an $s$-$v$ path also avoiding $y_j$, contradicting the fact that $y_j$ $s$-dominates $v$.
\end{proof}

\begin{proof}[Proof of \Cref{lem:maximal-univocal-paths-disjoint}]
    Since $p_1$ and $p_2$ are distinct we can assume without loss of generality that $p_1$ has a vertex not contained in $p_2$.
    Suppose for a contradiction that there is a vertex $v$ in $p_1$ and $p_2$.
    
    Let $u$ be the vertex in $p_1$ that is not contained in $p_2$ and that is closest to $v$ in $p_1$ (notice that the distance between two vertices of the same univocal path is identical in both trees, as well as  in $G$).
    Then vertex $u$ is either the immediate $s$-dominator of $v$ or it is immediately $s$-dominated by $v$. Suppose that $u$ is the immediate $s$-dominator of $v$, so in particular the vertex of $p_2$ closest to $s$ in the $s$-dominator tree is $v$.
    Since $p_1$ is univocal, $u$ immediately $s$-dominates only $v$ and $v$ immediately $t$-dominates only $u$. Thus, the trivial path $u$ concatenated with $p_2$ is univocal because $p_2$ is univocal, contradicting the maximality of $p_2$. For the remaining case when $u$ is immediately $s$-dominated by $v$, identically we can deduce that $p_2$ concatenated with $u$ is univocal, a contradiction.
\end{proof}

\begin{proof}[Proof of \Cref{thm:safe-sequences-characterization-C}]
$(\Rightarrow)$ Suppose for a contradiction that $X$ is safe and that for every vertex $v \in C$, $X$ is not a subsequence of $\ext{v}$. Then we can build a $C$-walk cover of $G$ by adding to it an $s$-$t$ walk covering each vertex in $C$ while avoiding $X$. This contradicts the $C$-safety of $X$.

$(\Leftarrow)$ If there is a vertex $v \in C$ such that $X$ is a subsequence of $\ext{v}$, then $X$ is safe. To see why, notice that $v$ has to be covered by some $s$-$t$ walk in any $C$-walk cover of $G$. This walk contains the sequence of $s$-dominators of $v$ followed by the sequence of $t$-dominators of $v$ by definition of dominance relation. Thus, $X$ is contained in such a walk and therefore it is safe.
\end{proof}

\begin{proof}[Proof of \Cref{thm:cores-vertices-C}]
    (For brevity we say ``dominator trees'' instead of ``blue-dominator trees'' and ``safe sequence'' instead of ``$C$-safe sequence''.)
    
    $(\Leftarrow)$ If $u$ is a leaf in both dominator trees then $\ext{u}$ is clearly maximal: no other vertex can be added to this sequence by the equivalence of safe sequences and paths in the dominator trees.
    
    $(\Rightarrow)$ Let $T_s$ and $T_t$ denote the $s$- and $t$- dominator trees. Suppose for a contradiction that $\ext{u}$ is a maximal safe sequence but $u$ is not a leaf in both dominator trees. Moreover, suppose without loss of generality that $u$ is not a leaf in $T_s$ and let $v$ be a blue-child of $u$ in $T_s$. Let $k\in\mathbb{N}^+$ be the smallest number such that $dom_s(v,k)=u$, so $u$ is the $k$-th ancestor of $v$ in $T_s$. Notice that $u \neq t$ since $u$ is an $s$-dominator. Let $w:=dom_t(u,k)$ (note that $w$ may or may not be in $C$).
    First we argue that $u$ has a blue-child $z$ in $T_s$ such that $z$ is strictly $t$-dominated by $w$ and $z$ is not a $t$-dominator of $u$, and later argue that $\ext{z}$ contradicts the maximality of $\ext{u}$. We do case analysis on the relation between $v$ and $w$.
    
    If $v \neq w$ then we can take $z = v$: by application of~\Cref{lem:st-dom-immediate-relation} to $u,v,k$ we have that $w$ $t$-dominates $v$ and moreover $v$ is not a $t$-dominator of $u$ since
    $v \neq w$, so $v$ is strictly $t$-dominated by $w$.
    
    Otherwise we have $v=w$. Vertex $v$ has a blue-child different from $u$ in $T_t$ or $u$ has a blue-child distinct from $v$ in $T_s$, otherwise $uv$ is a non-trivial $C$-univocal path, contradicting the fact that the trees are collapsed. So suppose without loss of generality that $u$ has a blue-child $v'$ in $T_s$ distinct from $v$. Let $k' \in\mathbb{N}^+$ be the smallest number such that $dom_s(v',k)=u$, so $u$ is the $k'$-th ancestor of $v'$ in $T_s$. Let $w':=dom_t(u,k')$. Applying~\Cref{lem:st-dom-immediate-relation} to $u,v',k'$ we have that $w'$ $t$-dominates $v'$ and $v'$ is not a $t$-dominator of $u$ unless $v'=w'$. Our goal now is to show that $v' \neq w'$. Suppose for a contradiction that $w' = v'$. We proceed by cases on the relation between $w$ and $w'$.
    
    If $w=w'$ then we have the following equalities $v=w=w'=v'$, which contradicts the fact that $v \neq v'$.
    
    Otherwise we have $w \neq w'$ and hence $w'$ is proper ancestor or proper descendant of $w$ in $T_t$, since they both $t$-dominate $u$. Suppose without loss of generality that $w'$ is a proper descendant of $w$ in $T_t$. Notice now that $v'$ does not $s$-dominate $v$ because $v$ is a blue-child of $u$, and so there is a $u$-$v$ path in $G$ avoiding $v'$; moreover, since $v=w$ and $v'=w'$ and $w'$ is a proper descendant of $w$ in $T_t$, $v'$ is a proper descendant of $v$ in $T_t$ and so there is a $v$-$t$ path in $G$ avoiding $v'$. These two paths concatenated at $v$ result in a $u$-$t$ path avoiding $v'$, contradicting the fact that $v'$ $t$-dominates by $u$. Therefore $v' \neq w'$. Now we can conclude that $w'$ strictly $t$-dominates $v'$. So $u$ has a blue-child $v'$, $w'$ $t$-dominates $u$, and $w'$ strictly $t$-dominates $v'$ and moreover $v'$ is not a $t$-dominator of $u$. We can thus \emph{put} $w=w'$ and take $z=v'$.
    
    We are now in conditions to argue that $\ext{z}$ properly contains $\ext{u}$ as a subsequence, which contradicts the maximality of $\ext{u}$ since $z$ is blue (i.e., $z\in C$) and so $\ext{z}$ is safe. 
    
    The sequence $\ext{u}$ consists of the sequence of $s$-dominators of $u$ (call it $Y_1$) concatenated with the sequence of $t$-dominators of $u$ (call it $Y_2$); analogously, the sequence $\ext{z}$ consists of the sequence of $s$-dominators of $z$ (call it $X_1$) concatenated with the sequence of $t$-dominators of $z$ (call it $X_2$).
    Clearly, since $z$ is a blue-child of $u$ in $T_s$, the sequence $Y_1$ is a proper subsequence of $X_1$.
    Moreover, since $z$ is not a $t$-dominator of $u$, either $u$ $t$-dominates $z$ or none $t$-dominates the other.
    In the former case $Y_2$ clearly is a proper subsequence of $X_2$ and thus $\ext{z}$ properly contains $\ext{u}$.
    In the latter case, if we show that $Y_2(u,w)$ is a subsequence of $X_1(u,z)$ we are done since $X_2$ contains $Y_2[w,t]$.\footnote{If $W$ is a sequence of dominators in a dominator tree then $W[a,b]$ denotes the subsequence of $W$ in between $a$ and $b$ and likewise for $W(a,b)$ except that $a$ and $b$ are removed.} First we argue that every vertex in $Y_2(u,w)$ is in $X_1(u,z)$. Suppose otherwise and let $x$ be such a vertex. Then there is a $u$-$z$ path avoiding $x$. Further, since $x \in Y_2(u,w)$ and $w$ is the lowest common ancestor between $u$ and $z$ in $T_t$, there is a $z$-$t$ path through $w$ avoiding $x$, which if concatenated at $z$ with the previous path yields a $u$-$t$ path avoiding $x$, contradicting the fact that $x$ $t$-dominates $u$.
    We are left to argue that the order is preserved. For that, let $x_1,\dots,x_\ell$ denote the sequence $Y_2(u,w)$ ($\ell \leq k-1$ by definition of $w$). Suppose for a contradiction that $x_1,\dots,x_\ell$ is not a subsequence of $X_1(u,z)$. Thus, with respect to $X_1(u,z)$, let $x_i$ denote the first vertex not matching the order of $x_1,\dots,x_\ell$ and let $x_j$ denote the vertex that should be in the position of $x_i$, so $i>j$. So $T_s$ witnesses that there is a $u$-$x_i$ path avoiding $x_j$ and $T_t$ witnesses that there is an $x_i$-$t$ path avoiding $x_j$ (as $i>j$), so $G$ has a $u$-$t$ path avoiding $x_j$, a contradiction to the fact that $x_j$ $t$-dominates $u$. This concludes the proof.
\end{proof}

\begin{proof}[Proof of \Cref{thm:representation-optimal-enumeration-C}]
    Build the dominator trees of $G$. From these, build the blue-dominator trees of $G$. Then collapse $C$-univocal paths (recall that we must store the ordered sequence of vertices of the collapsed path in the corresponding resulting node in one of the dominator trees). These three steps take $O(n+m)$ time, where $n$ denotes the number of vertices of $G$.
    
    For every vertex $v$ that is a leaf in both dominator trees, output $\ext{v}$. By \Cref{thm:cores-vertices-C} we find every maximal safe sequence during this process. Further, no safe sequence is computed twice since different leaves common to both trees have different extensions (as they differ in at least the vertex being extended). Therefore the algorithm runs in time equal to the length of all the maximal $C$-safe sequences plus the time required for preprocessing described above, so $O(m+o)$ altogether.
\end{proof}

\begin{figure}[h]
    \begin{subfigure}[m]{0.45\textwidth}
        \centering
        \includegraphics[width=0.8\textwidth]{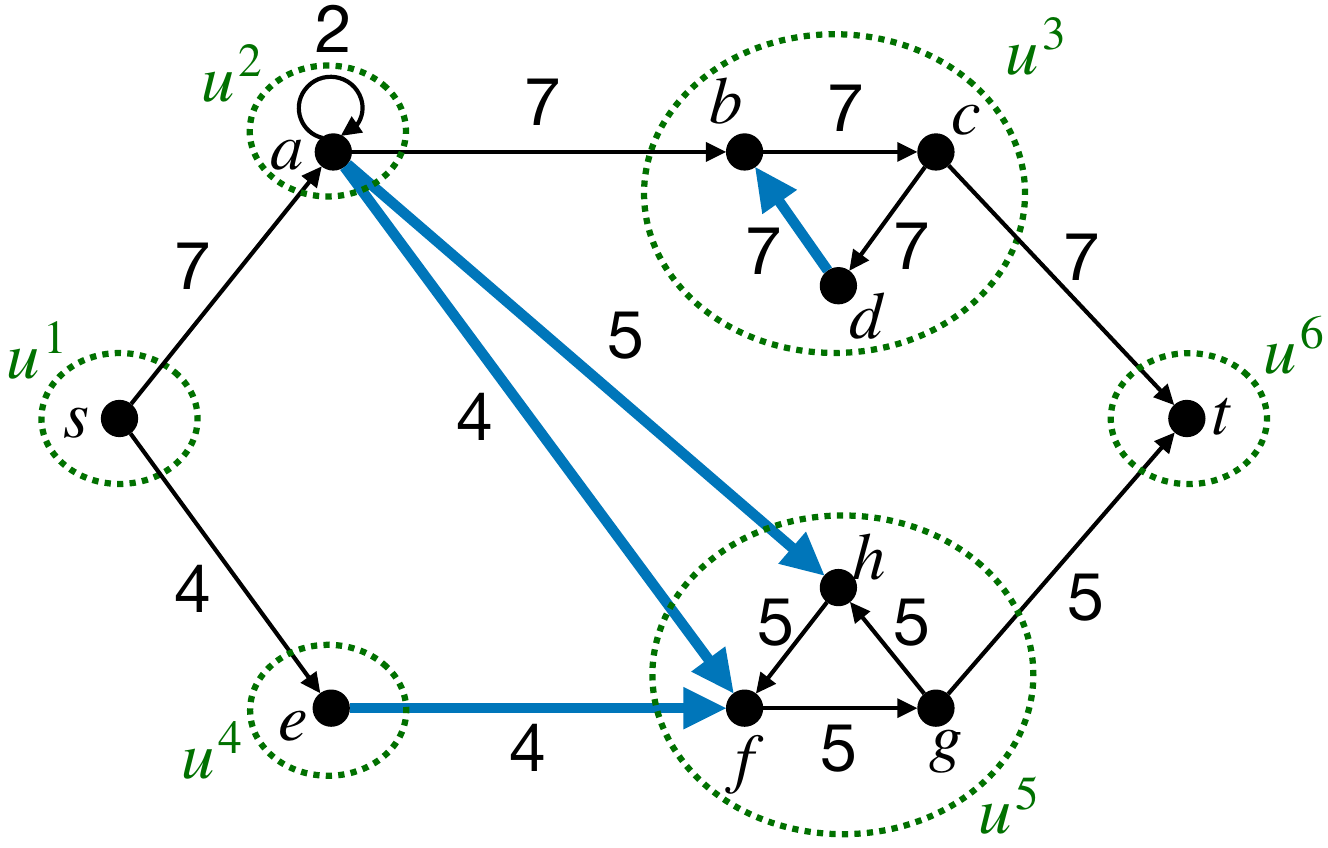}    
        \caption{An edge-weighted graph $G$, with its strongly connected components (SCC) circled in green, and labeled $u^1,\dots,u^6$.\label{fig:antichian-1}}
    \end{subfigure}
    \hfill
    \begin{subfigure}[m]{0.45\textwidth}
        \centering
        \includegraphics[width=0.8\textwidth]{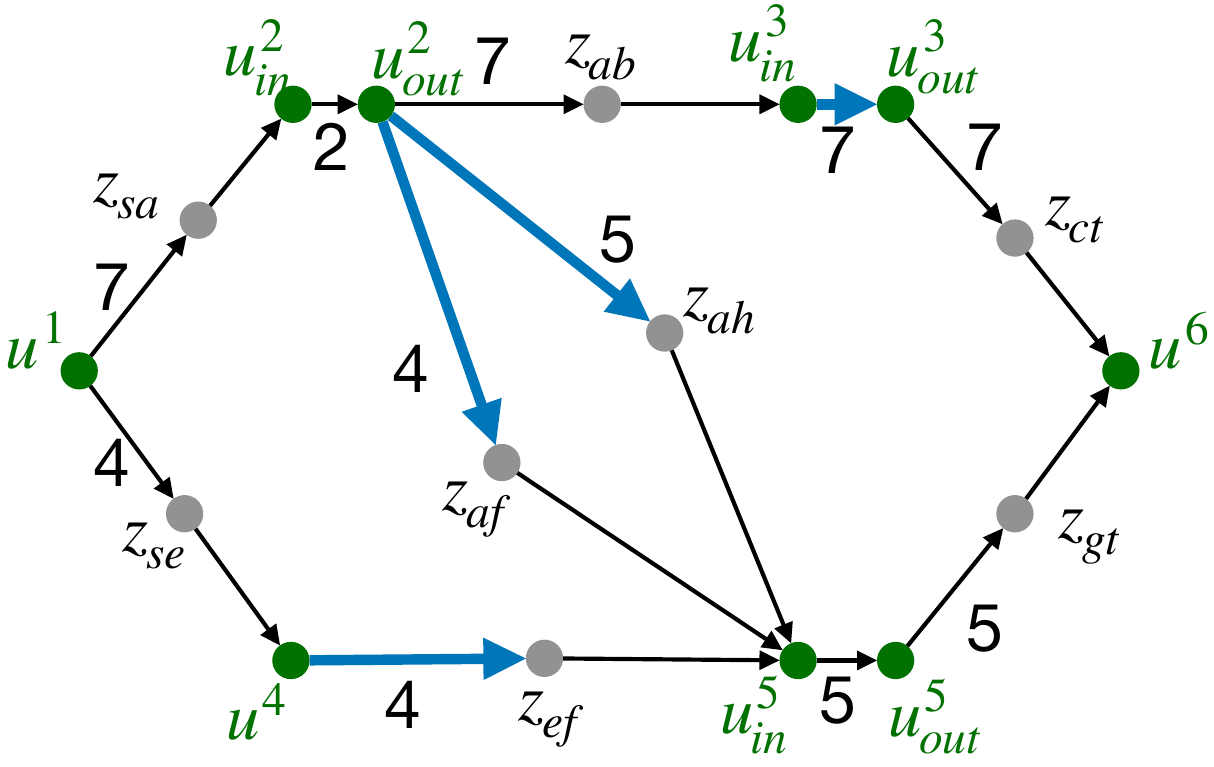}    
        \caption{The graph $G'$ from \Cref{thm:antichain-reduction}, with the correspondence between the SCCs of $G$ to the vertex gadgets of $G'$ also in green.}
    \end{subfigure}
    \caption{The four blue edges form a maximum-weight edge antichain of $G$, of weight 20. In $G'$, every non-trivial SCC $u^i$ of $G$ (i.e.~with at least one edge) is replaced by a gadget made up of a single edge $u^i_{in}u^i_{out}$ with weight equaling the maximum weight of an edge in that SCC. Every trivial SCC of $G$ (i.e. with no edges) is kept as a single vertex. All edges between SCCs are replaced by a path of length 2, with a private vertex (in gray) for every edge; the first edge of this length-2 path has the same weight as the original edge, and the second edge of this length-2 path has weight 0 (not shown). The maximum-weight edge antichain of $G'$, corresponding to the one shown for $G$, is also shown in blue.
    \label{fig:antichain-reduction}}
\end{figure}

\begin{proof}[Proof of \Cref{thm:antichain-reduction}]
    We build the following reduction (see also \Cref{fig:antichain-reduction} for an illustration). From $G$, we compute the acyclic graph $G^{scc} = (V^{scc},E^{scc})$ of strongly connected components (SCCs) of $G$ (the \emph{condensation} of $G$), in time $O(|V| + |E|)$. In this graph, the strongly connected components correspond to the vertices $V^{scc}$, and two vertices $v'$ and $u'$ in $V^{scc}$ are joined by an edge if there is an edge between a vertex in the SCC represented by $v'$ to a vertex in the SCC represented by $u'$. Moreover, we say that an SCC (or a vertex in $G^{scc}$) is non-trivial if the SCC contains at least one edge (note that a non-trivial SCC can also be made up of a single vertex with a self-loop). We build the reduction graph $G'$ from $G^{scc}$ as follows. If $v'$ is a non-trivial SCC, we represent $v'$ as an edge $v'_{in}v'_{out}$ in $G'$. In $G'$ we set the weight $w'(v'_{in}v'_{out})$ as the maximum weight of an edge in the SCC represented by $v'$. Otherwise, we simply add $v'$ to $G'$. For uniformity of future notation, in such a case we let $v'_{in} := v'$ and $v'_{out} := v'$. For every edge $uv$ in $G$ such that $u$ belongs to some SCC $u'$ and $v$ belongs to a different SCC $v'$, we add the 3-vertex path $u'_{out}z_{uv}v'_{in}$, where $z_{uv}$ is a new vertex private to the edge $uv$, which we add to avoid the creation of parallel edges. In $G'$, we set $w'(u'_{out}z_{uv}) = w(uv)$, and $w'(z_{uv}v'_{in}) = 0$; namely, the weight of $uv$ is carried over to the weight of the first edge of the newly created path in $G'$. The graph $G'$ can be computed in linear time, since $G^{scc}$ is computable in linear time~\cite{tarjan1972depth}, and the construction of $G'$ involves only a linear pass over $V$ and $E$. Note that the number of vertices of $G'$ is at most $2|V|$ and the number of edges is at most $|E|$.

    To prove the correctness of the reduction, we prove first that for any antichain $A$ of $G$, there exists an antichain $A'$ of $G'$ whose weight is at least the weight of $A$. Let $A'$ be a set of edges of $G'$ obtained from $A$ as follows:
    \begin{itemize}
        \item If $e \in A$ is an edge inside an SCC $v'$, then $v'$ is a non-trivial SCC and we take in $A'$ the edge $v'_{in}v'_{out}$. Since $w'(v'_{in}v'_{out})$ was chosen as the maximum weight of an edge in the SCC represented by $v'$, we have that $w'(v'_{in}v'_{out}) \geq w(e)$.
        \item Otherwise, $e = uv \in A$ is an edge between two SCCs, say $u'$ and $v'$, and we take in $A'$ the first edge of the corresponding 3-vertex path between $u'$ and $v'$, namely $u'_{out}z_{uv}$. By construction, its weight satisfies $w'(u'_{out}z_{uv}) = w(uv) = w(e)$. 
    \end{itemize} 
    By the above construction, we have that that the weight of $A'$ is at least the weight of $A$. It remains to show that $A'$ is an antichain in $G'$. Suppose for a contradiction this is not the case, and suppose that some edge $e_i' \in A'$ (created from an edge $e_i \in A$) reaches some edge $e_j' \in A'$ (created from an edge $e_j \in A$). However, by construction, we have that also $e_i$ reaches $e_j$, which contradicts the fact that $A$ is an antichain.
    
    Second, it remains to prove that for any antichain $A'$ of $G$, there exists an antichain $A$ of $G$ whose weight is at least that of $A'$. As above, from $A'$ we construct $A$ in a symmetric manner. First, we remove from $A'$ all edges whose weight is 0, as they do not change the weight of $A$. Next, we construct $A$ as follows.
    \begin{itemize}
        \item If $e' \in A'$ is an edge of the type $e' = v'_{in}v'_{out}$, then we include in $A$ an arbitrary edge $e$ inside the SCC corresponding to $v'$ of maximum weight inside that SCC. By construction, we have that $w(e) = w'(v'_{in}v'_{out}) = w'(e')$.
        \item If $e'$ is an edge of the type $e' = u'_{out}z_{uv}$, then we include in $A$ the edge $uv$. By construction, we have that $w(uv) = w'(u'_{out}z_{uv}) = w'(e')$.
    \end{itemize}
    Note that edges of $G'$ of the type $z_{uv}v'_{in}$ have weight 0, so they cannot belong to $A'$ anymore. Thus, the weight of $A$ equals the weight of $A'$. To see that $A$ is an antichain in $G$, note that if edges $e_i$ and $e_j$ of $A$ are such that there is a path from $v_i$ to $u_j$ in $G$, then also by construction there is a path in $G'$ between the corresponding edges $e_i'$ and $e_j'$, which contradicts the fact that $A'$ is an antichain in $G'$.
\end{proof}

\newpage
\section{Additional experimental results}
\label{sec:additional-experiments}

\begin{table}[h!]
\caption{\textbf{Experimental results on imperfect data with \ILPLAE.} The header is as in \Cref{tab:mfd}.\label{tab:lae}}
\centering
\scalebox{0.85}{
\begin{tabular}{|r|r|r|r|r|r|r|r|r|}
\hline
& \multirow{2}{*}{\#gen} 
& \multirow{2}{*}{\#graphs} 
& \multirow{2}{*}{\shortstack{avg $n$\\(max $n$)}} 
& \multirow{2}{*}{\shortstack{avg $m$\\(max $m$)}} 
& \multirow{2}{*}{prep (s)} 
& \multicolumn{2}{c|}{\#solved, time (s)} 
& \multirow{2}{*}{speed-up $(\times$)} \\ \cline{7-8}
& & & & & & no safety & safety & \\ \hline

\multirow{7}{*}{\rotatebox{90}{\footnotesize\textbf{complex32}}}
& 5 & 63 & 26 (144) & 37 (194) & 0.012 & 45, 44.66 & 59, 12.93 & 95.8 \\
& 10 & 162 & 28 (151) & 44 (209) & 0.020 & 106, 69.09 & 152, 14.17 & 249.3 \\
& 15 & 114 & 33 (156) & 54 (219) & 0.031 & 26, 185.00 & 105, 15.17 & 560.2 \\
& 20 & 116 & 42 (226) & 69 (318) & 0.046 & 4, 191.93 & 107, 17.39 & 411.2 \\
& 25 & 118 & 70 (284) & 111 (406) & 0.074 & 3, 144.66 & 91, 17.72 & 168.8 \\
& 30 & 122 & 77 (289) & 124 (416) & 0.100 & 2, 145.15 & 93, 27.22 & 102.1 \\
& 32 & 124 & 80 (291) & 131 (420) & 0.110 & 2, 226.23 & 90, 28.33 & 97.6 \\

\hline

\multirow{6}{*}{\rotatebox{90}{\footnotesize\textbf{ecoli}}}
& 5 & 127 & 11 (108) & 17 (157) & 0.006 & 104, 27.47 & 126, 5.48 & 111.1 \\
& 10 & 209 & 30 (249) & 46 (370) & 0.014 & 53, 135.49 & 169, 17.74 & 401.5 \\
& 15 & 271 & 53 (334) & 82 (502) & 0.031 & 4, 137.49 & 195, 31.57 & 364.4 \\
& 20 & 369 & 84 (398) & 128 (614) & 0.058 & 1, 9.37 & 236, 38.81 & 185.3 \\
& 25 & 436 & 99 (438) & 151 (672) & 0.070 & 0, - & 180, 75.12 & 28.6 \\
& 30 & 487 & 139 (466) & 210 (721) & 0.106 & 0, - & 163, 68.33 & 21.9 \\ 

\hline

\multirow{5}{*}{\rotatebox{90}{\footnotesize\textbf{JGI}}}
& 5 & 47 & 16 (86) & 25 (127) & 0.007 & 29, 63.04 & 45, 9.13 & 123.8 \\
& 10 & 41 & 34 (122) & 52 (175) & 0.018 & 3, 112.00 & 33, 16.05 & 387.6 \\
& 15 & 44 & 41 (132) & 66 (194) & 0.033 & 0, - & 35, 36.96 & 259.5 \\
& 20 & 42 & 53 (151) & 87 (224) & 0.057 & 1, 9.79 & 34, 38.19 & 111.3 \\
& 26 & 38 & 67 (181) & 111 (272) & 0.089 & 1, 92.37 & 27, 50.20 & 47.1 \\

\hline

\multirow{4}{*}{\rotatebox{90}{\footnotesize\textbf{medium20}}}
& 5 & 52 & 16 (75) & 24 (108) & 0.007 & 36, 43.28 & 49, 10.15 & 113.4 \\
& 10 & 39 & 26 (86) & 42 (131) & 0.014 & 7, 159.33 & 32, 19.05 & 450.8 \\
& 15 & 40 & 36 (113) & 60 (170) & 0.028 & 1, 247.94 & 31, 37.66 & 262.5 \\
& 20 & 40 & 50 (125) & 84 (200) & 0.069 & 0, - & 25, 49.29 & 159.3 \\

\hline
\end{tabular}}
\end{table}

\begin{table}[h!]
\caption{\textbf{Experimental results on imperfect data with \ILPRobust.} The header is as in \Cref{tab:mfd}.\label{tab:mpe}}
\centering
\scalebox{0.85}{
\begin{tabular}{|r|r|r|r|r|r|r|r|r|}
\hline
& \multirow{2}{*}{\#gen} 
& \multirow{2}{*}{\#graphs} 
& \multirow{2}{*}{\shortstack{avg $n$\\(max $n$)}} 
& \multirow{2}{*}{\shortstack{avg $m$\\(max $m$)}} 
& \multirow{2}{*}{prep (s)} 
& \multicolumn{2}{c|}{\#solved, time (s)} 
& \multirow{2}{*}{speed-up $(\times$)} \\ \cline{7-8}
& & & & & & no safety & safety & \\ \hline

\multirow{7}{*}{\rotatebox{90}{\footnotesize\textbf{complex32}}}
& 5 & 63 & 26 (144) & 37 (194) & 0.012 & 56, 15.81 & 63, 0.61 & 143.0 \\
& 10 & 162 & 28 (151) & 44 (209) & 0.017 & 41, 120.60 & 156, 2.85 & 1465.6 \\
& 15 & 114 & 33 (156) & 54 (219) & 0.029 & 6, 128.93 & 109, 4.09 & 965.3 \\
& 20 & 116 & 42 (226) & 69 (318) & 0.048 & 2, 193.06 & 112, 8.32 & 619.5 \\
& 25 & 118 & 70 (284) & 111 (406) & 0.085 & 0, - & 110, 14.02 & 277.5 \\
& 30 & 122 & 77 (289) & 124 (416) & 0.116 & 0, - & 110, 19.70 & 189.5 \\
& 32 & 124 & 80 (291) & 131 (420) & 0.122 & 0, - & 106, 18.28 & 173.0 \\

\hline

\multirow{6}{*}{\rotatebox{90}{\footnotesize\textbf{ecoli}}}
& 5 & 127 & 11 (108) & 17 (157) & 0.005 & 118, 14.85 & 127, 0.53 & 126.1 \\
& 10 & 209 & 30 (249) & 46 (370) & 0.017 & 21, 131.66 & 207, 4.07 & 1228.5 \\
& 15 & 271 & 53 (334) & 82 (502) & 0.039 & 1, 250.74 & 260, 7.25 & 694.4 \\
& 20 & 369 & 84 (398) & 128 (614) & 0.072 & 0, - & 336, 17.70 & 361.4 \\
& 25 & 436 & 99 (438) & 151 (672) & 0.094 & 0, - & 377, 19.75 & 177.4 \\
& 30 & 487 & 139 (466) & 210 (721) & 0.148 & 0, - & 390, 23.22 & 127.4 \\

\hline

\multirow{5}{*}{\rotatebox{90}{\footnotesize\textbf{JGI}}}
& 5 & 47 & 16 (86) & 25 (127) & 0.010 & 33, 32.78 & 46, 1.10 & 330.4 \\
& 10 & 41 & 34 (122) & 52 (175) & 0.021 & 0, - & 39, 6.79 & 840.5 \\
& 15 & 44 & 41 (132) & 66 (194) & 0.043 & 0, - & 39, 2.43 & 487.7 \\
& 20 & 42 & 53 (151) & 87 (224) & 0.057 & 0, - & 39, 5.91 & 277.6 \\
& 26 & 38 & 67 (181) & 111 (272) & 0.102 & 0, - & 33, 7.50 & 174.2 \\

\hline

\multirow{4}{*}{\rotatebox{90}{\footnotesize\textbf{medium20}}}
& 5 & 52 & 16 (75) & 24 (108) & 0.007 & 42, 22.17 & 52, 1.38 & 170.4 \\
& 10 & 39 & 26 (86) & 42 (131) & 0.019 & 0, - & 39, 9.27 & 1129.3 \\
& 15 & 40 & 36 (113) & 60 (170) & 0.033 & 0, - & 39, 5.32 & 537.0 \\
& 20 & 40 & 50 (125) & 84 (200) & 0.052 & 0, - & 37, 18.03 & 273.7 \\

\hline
\end{tabular}}
\end{table}

\newpage
\section{Flowpaths example code}
\label{sec:flowpaths}

The models in this paper were implemented in the \flowpaths Python package version 0.2.11, which is installable as:

\begin{lstlisting}
pip install flowpaths
\end{lstlisting}

Below, we give minimal usage examples of our implementation, to illustrate their use in potential applications employing these models. Note that the decomposition models in \flowpaths accept graph object created using the popular NetworkX Pyton library. Moreover, the \flowpaths package uses the \texttt{graphviz} Python package to visualize graphs, which can be useful in applications. We also show these visualizations for the code fragments below. The complete documentation and further examples of the classes implementing our models is available in the \flowpaths documentation, at \url{algbio.github.io/flowpaths}.

\subsection{Minimum Flow Decomposition}

\begin{lstlisting}
import flowpaths as fp
import networkx as nx

# We create a NetworkX graph with a flow value on each edge
graph = nx.DiGraph()
graph.add_edge('s', 'a', flow=4)
graph.add_edge('a', 'a', flow=1)
graph.add_edge('a', 'b', flow=3)
graph.add_edge('a', 'c', flow=1)
graph.add_edge('b', 'c', flow=6)
graph.add_edge('c', 'd', flow=8)
graph.add_edge('d', 'e', flow=1)
graph.add_edge('d', 'f', flow=6)
graph.add_edge('d', 't', flow=1)
graph.add_edge('e', 'c', flow=1)
graph.add_edge('f', 'b', flow=3)
graph.add_edge('f', 't', flow=3)

# We visualize the input graph
fp.utils.draw(
    G = graph,
    filename = 'input-mfd.pdf',
    flow_attr = 'flow',
    draw_options={'show_edge_weights': True})
\end{lstlisting}

The input graph rendered as a PDF file by the above code is the following:
\vspace{-0.9cm}
\begin{figure}[h!]
    \centering
    \includegraphics[width=0.8\linewidth]{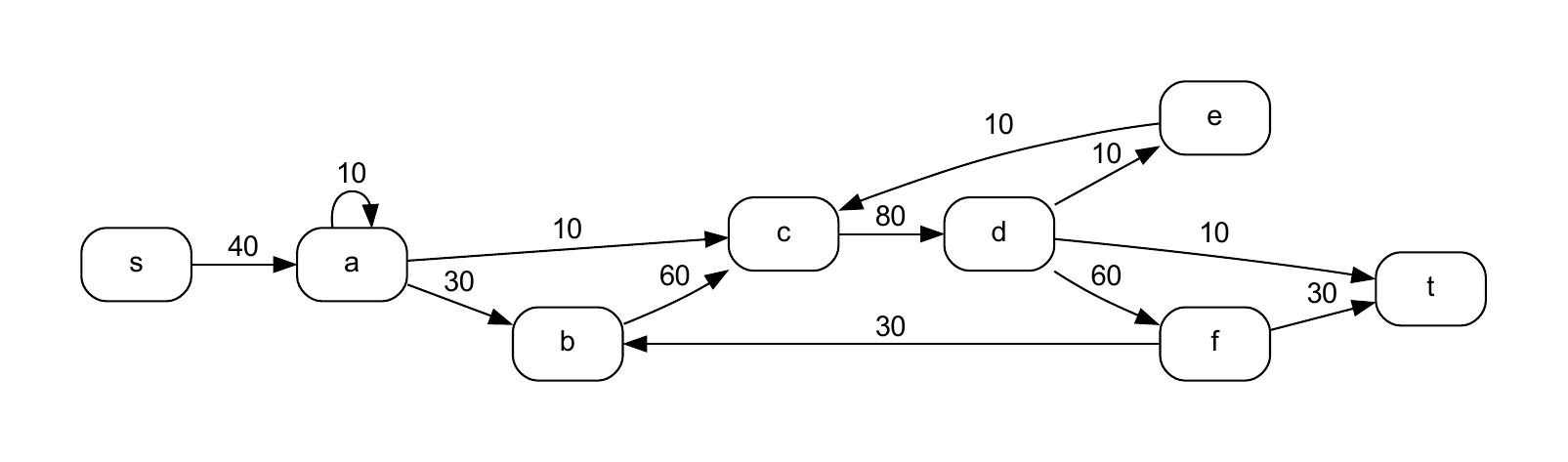}
\end{figure}
\vspace{-0.7cm}

Next, we proceed to creating and solving the model:
\begin{lstlisting}
# We create the model
model = fp.MinFlowDecompCycles(
    G = graph,
    flow_attr = 'flow')

# We solve the model
model.solve()

# If the model was solved, we retrieve the solution
if model.is_solved():
    solution = model.get_solution()
    print('Solution walks:', solution['walks'])
    print('Solution weights:', solution['weights'])
    fp.utils.draw( # We visualize the solution
        G = graph,
        filename = 'solution-mfd.pdf',
        paths=solution['walks'],
        weights=solution['weights'])
\end{lstlisting}

The solution walks visualized by the above code are the blue and red walks in the figure below, of weights 10 and 30, respectively. With the default drawing options, the weight is printed on the first edge of the walk. Note also that e.g. the blue walk traverses the edge $(c,d)$ two times.

\vspace{-0.9cm}
\begin{figure}[h!]
    \centering
    \includegraphics[width=0.8\linewidth]{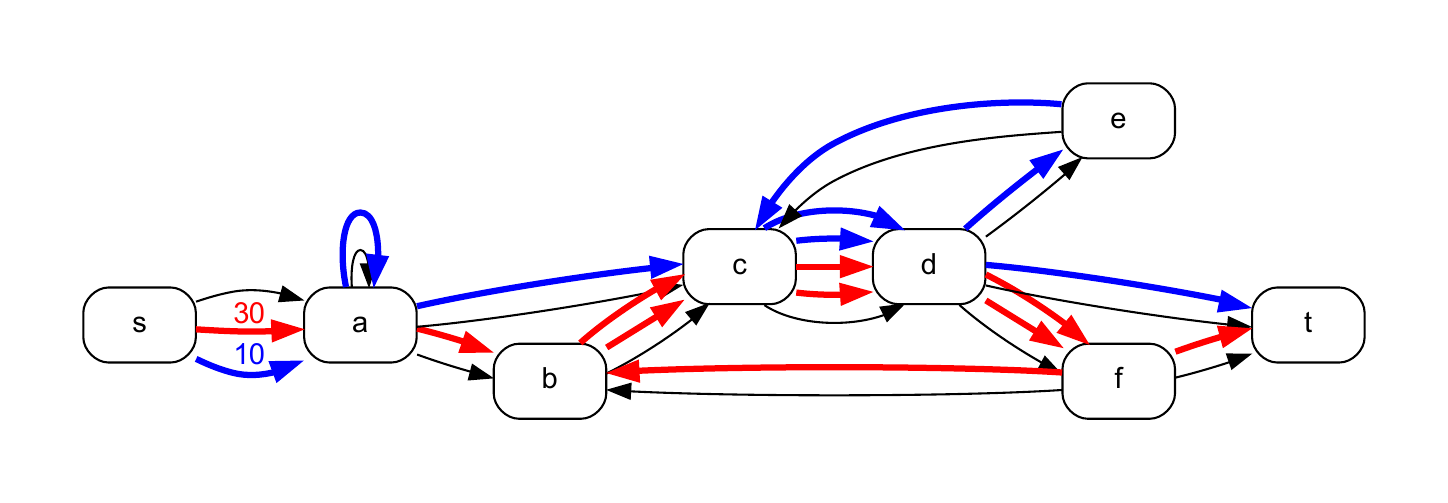}
\end{figure}
\vspace{-0.5cm}

\subsection{Least Absolute Errors}

For the \ILPLAE model, we illustrate the use of both subset constraints, and of enforcing the objective from \Cref{def:problems} on a proper subset $E' \subsetneq E$ of edges. 

Assume that the input graph is the one illustrated below, where there are two subset constraints \texttt{[('s', 'a'), ('a', 'c'), ('d', 't')]} and \texttt{[('f', 'b'), ('b', 'c'), ('d', 'f')]}, visualized as dotted colored edges below.





\vspace{-0.9cm}
\begin{figure}[h!]
    \centering
    \includegraphics[width=0.8\linewidth]{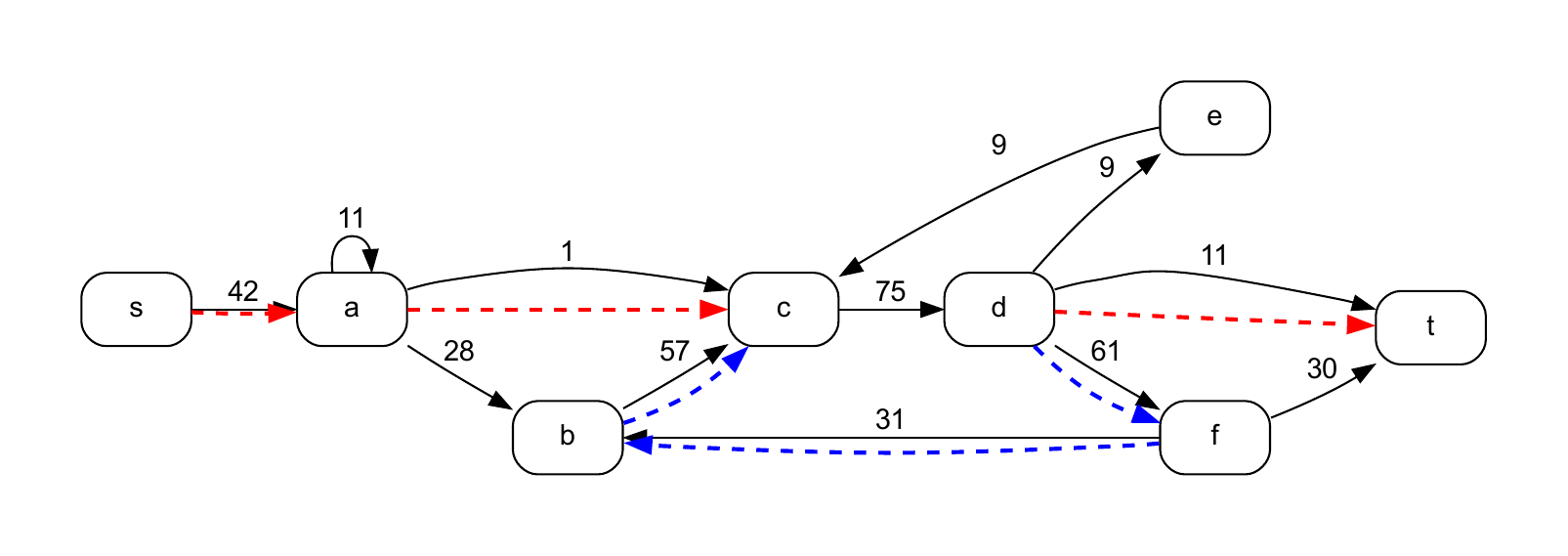}
\end{figure}
\vspace{-0.7cm}

We can pass such subset constraints via the parameter \texttt{subset\_constraints}. Moreover, since the flow value on edge \texttt{('a', 'c')} is much smaller than the other values, one could remove it from the set $E'$. This is achieved by passing such edges to the parameter \texttt{elements\_to\_ignore}.

\begin{lstlisting}
# We create the model
model = fp.kLeastAbsErrorsCycles(   
    G = graph,
    weight_type = int,
    flow_attr = 'flow',
    subset_constraints = [
        [('s', 'a'), ('a', 'c'), ('d', 't')],
        [('b', 'c'), ('d', 'f')]
    ],
    elements_to_ignore = [('a', 'c')])

# We solve the model
model.solve()

# If the model was solved, we retrieve the solution
if model.is_solved():
    solution = model.get_solution()
    print('Solution walks:', solution['walks'])
    print('Solution weights:', solution['weights'])
\end{lstlisting}

The solution walks are visualized in the PDF figure below.

\vspace{-0.9cm}
\begin{figure}[h!]
    \centering
    \includegraphics[width=0.8\linewidth]{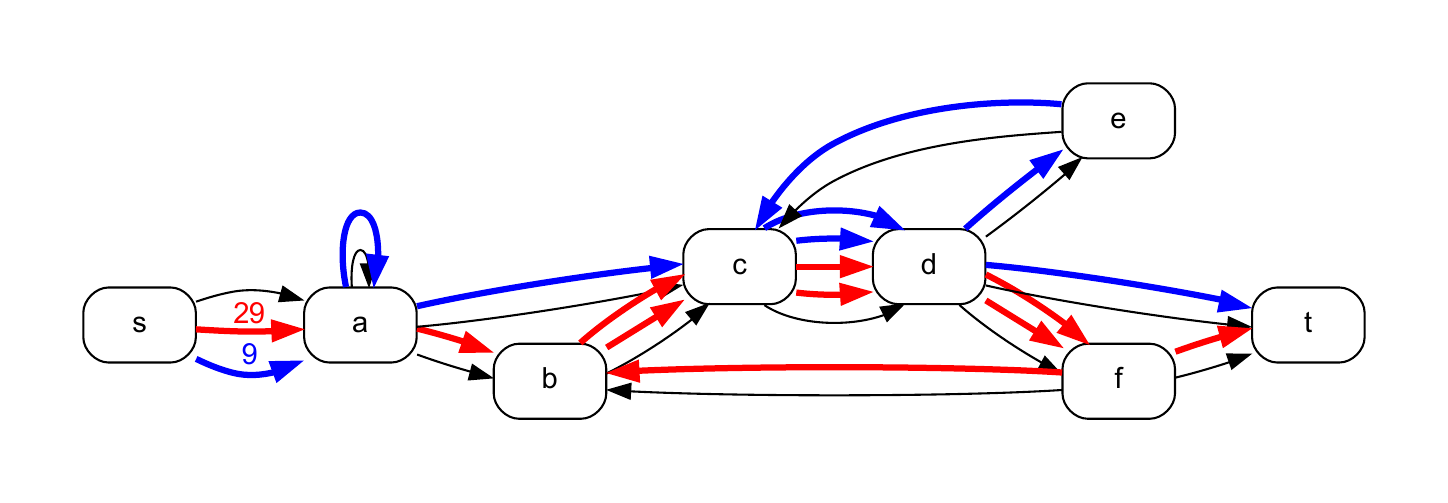}
\end{figure}
\vspace{-0.5cm}

\end{document}